\documentclass[11pt,aps,pra,onecolumn,superscriptaddress,floatfix,nofootinbib,showpacs,longbibliography]{revtex4-2}
\usepackage[utf8]{inputenc}  
\usepackage[T1]{fontenc}     
\usepackage[british]{babel}  
\usepackage[sc,osf]{mathpazo}
\usepackage{libertineRoman}  
\usepackage[colorlinks=true, citecolor=blue, urlcolor=blue]{hyperref}  
\usepackage{graphicx}
\usepackage{diagbox}
 \usepackage{enumitem}
 \usepackage[babel]{microtype}  
\usepackage{amsmath,amssymb,amsthm,bm,amsfonts,mathrsfs,bbm} 

\usepackage{xspace}  
\usepackage{pgf,tikz}
\usepackage{xcolor}
\usepackage{multirow}
\usepackage{array}
\usepackage{bigstrut}
\usepackage{braket}
\usepackage{color}
\usepackage{natbib}
\usepackage{multirow}
\usepackage{mathtools}
\usepackage{float}
\usepackage{xcolor,colortbl}
\usepackage{physics}
\usepackage{amsmath}
\usepackage{color}
\usepackage[justification=justified, format=plain]{subcaption}
\usepackage[justification=raggedright]{caption}
\usepackage[normalem]{ulem}

\newcommand{\be}{\begin{equation}}
\newcommand{\ee}{\end{equation}}
\newcommand{\ba}{\begin{eqnarray}}
\newcommand{\ea}{\end{eqnarray}}

\newtheorem{theorem}{Theorem}
\newtheorem{corollary}{Corollary}[theorem]
\newtheorem{definition}{Definition}
\newtheorem{proposition}{Proposition}

\newtheorem{remark}{Remark}
\newtheorem{lemma}{Lemma}

\newtheorem{conjecture}{Conjecture}
\DeclareMathOperator{\id}{id}

\newcommand{\biid}[1]{\operatorname{BIID}_{#1}}

\def\>{\rangle}
\def\<{\langle}

\begin{document}

\title{Quantum hypothesis testing of non-mixed-unitarity: \\A multifaceted hierarchy of quantum channel discrimination}

\author{Pratik Ghosal}
\affiliation{Harish-Chandra Research Institute, Chhatnag Road, Jhunsi, Prayagraj - 211019, India}
\affiliation{Homi Bhabha National Institute, Training School Complex, Anushakti Nagar, Mumbai 400 094, India}

\author{Pritam Halder}
\affiliation{Harish-Chandra Research Institute, Chhatnag Road, Jhunsi, Prayagraj - 211019, India}
\affiliation{Homi Bhabha National Institute, Training School Complex, Anushakti Nagar, Mumbai 400 094, India}
\affiliation{Networked Quantum Devices Unit, Okinawa Institute of Science and Technology Graduate University, Okinawa, Japan}

\author{Ayan Patra}
\affiliation{Harish-Chandra Research Institute, Chhatnag Road, Jhunsi, Prayagraj - 211019, India}
\affiliation{Homi Bhabha National Institute, Training School Complex, Anushakti Nagar, Mumbai 400 094, India}
\affiliation{IT:U Interdisciplinary Transformation University, Freistädter Strasse 400, 4040 Linz, Austria}

\author{Aditi Sen(De)}
\affiliation{Harish-Chandra Research Institute, Chhatnag Road, Jhunsi, Prayagraj - 211019, India}
\affiliation{Homi Bhabha National Institute, Training School Complex, Anushakti Nagar, Mumbai 400 094, India}

\begin{abstract}
Given multiple uses of an unknown quantum channel, we determine whether it is a specified non-mixed-unitary channel or belongs to the set of mixed-unitary channels. Formulating this as a composite channel hypothesis testing problem, we characterize the Stein exponents achievable by parallel strategies under progressively weaker restrictions on probe states and auxiliary memory. In particular, we consider block-i.i.d. probes, which allow arbitrary correlations within blocks of fixed size while remaining i.i.d. across different blocks, and derive finite-letter expressions of the corresponding Stein exponents. Varying the block size, which interpolates between fully i.i.d. and arbitrary probes, together with further restrictions on intra-block correlations and access to auxiliary memory, results in a multifaceted hierarchy. We establish several strict separations within this hierarchy. Without auxiliary memory, we prove that fully i.i.d. probes yield a vanishing Stein exponent for every unital non-mixed-unitary channel, while blocks of three product probes suffice to obtain a strictly positive exponent for the qutrit Werner-Holevo channel. For block size two, although product probes seem to be insufficient for $O(3)$-covariant channels, an entangled probe achieves a strictly positive exponent for the qutrit Werner-Holevo channel, even though every maximally entangled probe fails. In contrast, access to auxiliary memory makes fully i.i.d. probes sufficient by yielding a strictly positive exponent for every non-mixed-unitary channel, even without requiring input–reference entanglement, and this exponent admits a lower bound determined by the diamond-norm distance from the mixed-unitary set. Moreover, for odd-dimensional Werner-Holevo channels, we find that this exponent is infinite, with every full-Schmidt-rank pure state being optimal.

\end{abstract}

\maketitle

\section{Introduction}
\label{sec:intro}

Determining whether a given quantum state, channel, or, more generally, a quantum process possesses a particular quantum property is a central problem in quantum information theory \cite{eisert2020quantum, Friis2019, Berta2024, PhysRevLett.110.060405, PhysRevLett.121.180503, nlz1-h6qr, PhysRevLett.107.050502, PhysRevX.8.021033, bwd5-wx7j, xry9-j4g6, Krawiec2021, grinko2026sample}.
{It helps to identify quantum resources with appropriate properties that can provide desired advantages in various information-processing tasks.} Quantum channels describe the most general discrete-time evolution of quantum systems, including the effects of interactions with an environment, and therefore play a fundamental role in characterizing quantum devices and noise processes relevant to quantum information-processing tasks. {Consequently, distinguishing quantum channels and determining their specific properties have received significant attention in recent years. Beyond theoretical studies, quantum channel discrimination has also been experimentally demonstrated on experimental platforms such as photonic setups, where both unitary and nonunitary processes have been experimentally discriminated~\cite{PhysRevLett.102.160502, agresti2019experimental}, and on trapped ions~\cite{PhysRevLett.131.170602}.}

Mathematically, quantum channels are represented by linear completely positive and trace-preserving (CPTP) maps between the sets of linear operators acting on the complex Hilbert spaces of the initial and final systems \cite{watrous2018theory}. In the special case of an isolated quantum system,  the system's evolution is reversible and can thus be described by a \emph{unitary quantum channel} ($\mathcal{U}:\rho \mapsto \mathcal{U}(\rho)=U\rho U^{\dagger}$, where $U$ is a unitary operator), {whereas for open systems, the dynamics is, in general, irreversible}.
A natural generalization of unitary channels is provided by \emph{mixed-unitary channels}.
A quantum channel $\Phi$ is said to be mixed-unitary if it can be expressed as a convex combination of unitary channels, i.e., there exists a probability vector $(p_i)$ and an ensemble of unitary channels $\{\mathcal{U}_i\}$ such that $\Phi=\sum_i p_i\,\mathcal{U}_i$. Unlike general quantum channels, mixed-unitary channels admit a simple physical interpretation: they can be understood as arising from classical uncertainty regarding an underlying reversible evolution. Furthermore, they serve as a natural example of channels that preserve the maximally mixed state, known as {\it unital channels},
{which are} the quantum counterparts of doubly stochastic maps in classical probability theory.
In classical theory, Birkhoff's theorem 
guarantees that every doubly stochastic map can be expressed as a probabilistic mixture of reversible maps \cite{Birkhoff1946, Neumann1953}. The analogous statement for quantum channels, however, holds only in dimension two while in dimensions three or higher, mixed-unitary channels constitute a strict subset of unital channels \cite{Landau1993, Mendl2009, Haagerup2011, Haagerup2015} (see also \cite{Chiribella2017, chakraborty2024asymptotic}).


{The fact that the irreversibility of mixed-unitary dynamics, unlike that of non-mixed-unitary dynamics, originates from classical uncertainty has several distinctive operational consequences.}
For example, mixed-unitary dynamics allow for perfect recovery of quantum information leaked from a system to its environment using correction schemes based solely on classical outcomes of measurements performed on the environment, while  such recovery is not generally possible for non-mixed-unitary channels \cite{gregoratti2003quantum, gregoratti2004quantum}. This distinction also has important implications for the study of memory effects in non-Markovian quantum dynamics \cite{PhysRevLett.132.060402}. In particular, it has recently been shown that a two-step dynamics $\mathcal{D}=(\Phi,\mathrm{id})$ is realizable with a classical memory if and only if $\Phi$ is mixed-unitary \cite{dn6t-y9ky}. Consequently, whenever $\Phi$ is non-mixed-unitary, the dynamics necessarily require genuinely quantum memory. Furthermore, non-mixed-unitarity of the reduced system dynamics has recently been identified  as a witness of system-environment entanglement in pure-dephasing dynamics \cite{lin2025resource}. These results establish that non-mixed-unitarity is intimately connected to genuinely quantum features of open-system dynamics, motivating the problem of operationally discriminating non-mixed-unitary channels from their mixed-unitary counterparts, which is the focus of this work. 
Similar to the entanglement versus separability problem~\cite{RevModPhys.81.865}, despite the simple definition of mixed-unitary channels, determining whether a given quantum channel is mixed-unitary or not is far from straightforward. Since the set of mixed-unitary channels $\mathsf{MU}$ is convex and compact, the Hahn-Banach separation theorem guarantees that every non-mixed-unitary channel has a witness that distinguishes it from $\mathsf{MU}$ \cite{Mendl2009, PhysRevA.93.032132}. 
However, constructing such a witness for a given channel is itself a nontrivial problem, analogous to constructing an entanglement witness for a given quantum state \cite{RevModPhys.81.865, bruss2002reflections, guhne2009entanglement}. A complete characterization of the mixed-unitary region is known only for certain specific parametric families of channels \cite{Mendl2009}. In fact, the computational complexity of the general mixed-unitarity membership problem has been shown to be NP-hard \cite{Lee2020detectingmixed}.
More recently, semidefinite programming-based witnesses derived from the theory of quantum memory in non-Markovian dynamics have been proposed \cite{dn6t-y9ky}.

In this work, we formulate the problem of determining the non-mixed-unitarity of a given channel as a \emph{quantum hypothesis testing problem} and characterize the fundamental limits of discriminating a non-mixed-unitary channel from the set of mixed-unitary channels under various operational restrictions.


\subsection{Problem Statement}

Let $\mathcal{N}$ be a fixed non-mixed-unitary channel. Given $n$ independent and identical uses of an unknown quantum channel $\mathcal{E}$, we consider the problem of determining whether $\mathcal{E}$ is the specified channel $\mathcal{N}$ or instead belongs to the set of mixed-unitary channels $\mathsf{MU}$, i.e., 
the task is to distinguish between the following hypotheses:
\begin{align*}
&\text{Null hypothesis}~(H_0):\quad
\mathcal{E}^{\otimes n}\in\mathsf{MU}^{\mathrm{iid}}_n:=\left\{\Phi^{\otimes n}:\Phi\in\mathsf{MU}\right\},\\
\text{and}~~&\text{Alternative hypothesis}~(H_1):\quad\mathcal{E}^{\otimes n}=\mathcal{N}^{\otimes n}.
\end{align*}
This constitutes a binary composite quantum hypothesis testing problem in the independent and identically distributed (i.i.d.) setting \cite{Berta2021on,bergh2023composite}. The i.i.d. structure reflects the assumption that the same unknown channel is used independently in each instance. We are ultimately interested in the asymptotic limit $n\to\infty$.\\

\noindent\textbf{Setting.--} The above discrimination task can be approached using different strategies \cite{PhysRevLett.127.200504, fang2025towards}. One may employ a \emph{parallel strategy}, in which the $n$ channel uses are applied simultaneously to an $n$-partite probe state. Further, an auxiliary perfect side channel (or memory) of suitable dimension can be used, which acts on a reference system that can be correlated with the probe.  Throughout this paper, strategies in which the input probe is allowed to be correlated with an auxiliary reference system will be called \textit{memory-assisted strategies}, whereas those that do not use an additional reference system will be called \textit{memoryless strategies}. The resulting output state is then subjected to a two-outcome quantum measurement, whose outcome determines the decision between the hypotheses $H_0$ and $H_1$. Alternatively, one may employ \emph{adaptive strategies},\footnote{Adaptive strategy: An input probe state, potentially entangled with a reference system, is supplied to the first use of the channel, while the reference system is transmitted through the auxiliary side channel. The resulting output is then processed by an intermediate update (or control) channel before being fed into the next use of the channel. This procedure is repeated between successive channel uses, and the output obtained after the final channel use is subjected to a two-outcome quantum measurement.} which
are more general, containing parallel strategies as a special case \cite{PhysRevLett.127.200504}. However,  the parallel ones are considerably simpler to analyze. In particular, determining an optimal parallel strategy requires optimization only over the input probe state and the final measurement, whereas adaptive ones involve additional optimization over intermediate update channels. In this paper, we restrict our attention to parallel discrimination strategies, both with and without auxiliary memory.

Even within the class of parallel strategies, there may be additional restrictions on the admissible probe states. One natural restriction concerns the correlation structure across different channel uses.  The most general parallel strategy allows arbitrary probe states $\rho_{(n)}$ (potentially correlated with a reference system $R'$) as input across $n$ channel uses (see Fig.~\ref{fig:parallel}). However, in the asymptotic limit, these general probe states generally lead to a regularized expression for the Stein exponent (discussed subsequently) \cite{bergh2023composite}. We therefore consider probe states possessing a block-i.i.d. structure,
\begin{align}
    \rho_{(n)}=\sigma_{(k)}^{\otimes m},\qquad n=mk,\label{block}
\end{align}
where $\sigma_{(k)}$ is the probe state (potentially correlated with a reference system $R$ in each block) used as input across a block of $k$ channel uses. Such states permit arbitrary correlations (both classical and quantum) within each block of size $k$, while requiring the different blocks to be independent and identically distributed. This setting naturally captures scenarios in which the use of multipartite entanglement is restricted to finite-size blocks of channel uses, with different blocks remaining separable. The special case $k=1$ corresponds to fully i.i.d. probe states of the form $\rho^{\otimes n}$, whereas the case $k=n$, and hence in the asymptotic limit $k\to\infty$, recovers the most general arbitrary probe setting. One may further impose restrictions on the entanglement structure of the block state $\sigma_{(k)}$, like separability across suitable partitions.\\

\noindent\textbf{Figure of merit.--} In order to measure the success of any hypothesis testing task one typically computes two types of errors. A Type I error occurs when the null hypothesis $H_0$ is rejected even though it is true, whereas in the case of Type II error, the null hypothesis is accepted even though the alternative hypothesis $H_1$ is true \cite{HiaiPetz1991, OgawaNagaoka2000}. In the present setting, a Type I (false positive) error corresponds to concluding that the given channel is non-mixed-unitary when in fact it is mixed-unitary, while a Type II (false negative) error corresponds to concluding that the given channel is mixed-unitary when in fact it is non-mixed-unitary. Our aim is, therefore, to keep the false positive error probability below a prescribed threshold $\epsilon$, ideally close to zero, thereby avoiding incorrect detection of a mixed-unitary channel as non-mixed-unitary. At the same time, we want the testing protocol to be as powerful as possible, by minimizing the false negative error probability so as not to overlook genuine non-mixed-unitary channels. This naturally leads to the asymmetric hypothesis testing setting. Our central figure of merit is the \emph{Stein exponent}, which quantifies the optimal asymptotic exponential decay rate of the Type II error probability while keeping the Type I error probability at most $\epsilon$ for all $n$ \cite{HiaiPetz1991, OgawaNagaoka2000}.

\subsection{Summary of Results}

We first briefly discuss composite i.i.d. channel hypothesis testing under parallel discrimination strategy, introducing the block-i.i.d. probe setting. (Sec.~\ref{sec:compchannel}). Next, we develop the general framework for testing a specified channel against a closed convex set of channels using block-i.i.d. probes (Sec.~\ref{sec:block-iid}). For every fixed block size $k$, we derive expressions for the memoryless and memory-assisted Stein exponents, denoted by $E^{k}$ and $\widetilde{E}^{k}$, respectively, in terms of the output quantum relative entropy. Unlike the case of arbitrary probes, the corresponding Stein exponents for block-i.i.d. probes admit finite-letter characterizations for every fixed $k$. This framework bridges fully i.i.d. strategies and arbitrary probe strategies through regularization over the block size, and induces a hierarchy of probe strategies parametrized by the block size $k$. In particular, for every fixed $k\geq2$,
\(E^{1}\leq E^{k}\leq E\),
where $E$ denotes the Stein exponent corresponding to arbitrary probes, and the same relation holds in the memory-assisted setting. Moreover, for every fixed block size $k$, \(E^{k}\leq\widetilde{E}^{k}\). We then consider further restrictions on the correlations between channel inputs within each block. In particular, we compute the Stein exponents for the cases where the probes within each block are arbitrarily correlated with no restrictions (\(E^{k}\)), fully separable (\(E^k_{\text{Sep}}\)), and product with no classical correlations as well (\(E^k_{\text{Pro}}\)). These give rise to a further hierarchy,
\(E_{\mathrm{Pro}}^{k}\leq E_{\mathrm{Sep}}^{k}\leq E^{k}\),
with analogous relations for the corresponding memory-assisted exponents.

We then specialize the framework to the case of a mixed-unitary null hypothesis (see Sec.~\ref{sec:MU-hypothesis}). Taking a resource-theoretic perspective in which the set of mixed-unitary channels is the free set, we define a simple class of free superchannels. We prove that all the exponents introduced above are nonincreasing under the action of this class of superchannels. Thus, these exponents serve as operational resource monotones of non-mixed-unitarity.

We examine the multifaceted hierarchy 
for the case of non-mixed-unitary channels (Sec.~\ref{sec:MU-hierarchy}). We present instances where these hierarchies are strict. In  the memoryless scenario, we prove that fully i.i.d. probes yield a vanishing Stein exponent for every non-mixed-unitary unital channel. Increasing the block size to two, we find that product probes seem to be insufficient to give a nonzero exponent for $O(3)$-covariant channels---which constitute a subset of qutrit unital channels. In particular, we analytically prove that for several broad classes of qutrit-pairs $\{\rho_1,\rho_2\}$, there exists a pair-specific mixed-unitary channel $\Phi$, satisfying \(\Phi(\rho_i)=\mathcal N(\rho_i),~i=1,2,\) for every $O(3)$-covariant channel $\mathcal{N}$. Numerical evidence further supports that this simultaneous simulation holds for arbitrary pairs of qutrit states. We therefore conjecture that the memoryless Stein exponent for product probes of block size two vanishes for $O(3)$-covariant channels. However, for the qutrit Werner-Holevo channel~\cite{Werner2002}, which is a particular $O(3)$-covariant channel, we prove that blocks containing three product probes yield a positive Stein exponent, thereby establishing a strict advantage of increasing the block size.

Entanglement between channel inputs (Sec. \ref{subsec:entanglement_advantage}) and/or access to auxiliary memory (Sec. \ref{subsec:memory_advantage}) can provide other possible routes to overcoming the apparent limitations of block-i.i.d. probes with blocks of size two. In the former memoryless case, we find that a Schmidt-rank-two entangled probe yields a positive Stein exponent
for the qutrit Werner--Holevo channel. Conditional on the conjecture of vanishing Stein exponent for product probes of block size two, this establishes the strict separation between product probes and entangled probes at block size two.
However, we also show that every maximally entangled two-qutrit probe gives a vanishing exponent for the same Werner-Holevo channel, thereby indicating that the advantage of entanglement does not scale with the amount of entanglement. 
In the memory-assisted case,  we prove that, fully i.i.d. memory-assisted probes achieve a strictly positive Stein exponent for every non-mixed-unitary channel with a lower bound given by the diamond-norm distance between the channel and the mixed-unitary set.  
This establishes the strict advantage of auxiliary memory, which, interestingly, does not require entanglement between the input and reference systems: we prove that suitably chosen separable but correlated probes can also provide a strictly positive Stein exponent. Moreover, for odd-dimensional Werner-Holevo channels and their unitary equivalents, we determine the exact Stein exponent in the memory-assisted fully i.i.d. setting, 
which turns out to be infinite with every pure input--reference state of full Schmidt rank serving as an optimal probe.

\section{Preliminaries}
\label{sec:prelim}
We begin by fixing some notations in this section which are extensively used throughout the paper. We then discuss few properties of quantum state and channel divergences relevant to our analysis, followed by a brief overview of composite quantum state hypothesis testing.

\subsection{Notation}
We denote quantum systems by $A,B,R,\ldots$, with the corresponding finite-dimensional complex Hilbert spaces denoted by the same symbols. We also use $\mathbb{C}^d$ to specify a $d$-dimensional complex Hilbert space.  The set of all linear operators and density matrices acting on $A$ are denoted by $\mathfrak B(A)$ and $\mathfrak D(A)$, respectively. A quantum channel $\mathcal E$ is an element of the set of linear, completely positive and trace preserving (CPTP) maps from $\mathfrak B(A)$ to $\mathfrak B(B)$, simply denoted by $\mathsf{CPTP}(A\to B)$. Throughout, `$\log$' refers to the logarithm to base $2$, while `$\ln$' represents logarithm to the base $e$.



\subsection{Quantum State and Channel Divergences} 
Let us first recall two quantum state divergences that will play a central role in our work.

\begin{definition}
\label{def:sd}  
    Let $\rho\in\mathfrak D(A)$ be a quantum state and $\sigma\in\mathfrak{B}(A)$ be a positive semi-definite operator.
    \begin{itemize}
        \item The \textup{Umegaki relative entropy} (also known as \textup{quantum relative entropy}) \cite{Watrous2011} of $\rho$ with respect to $\sigma$ is defined as
\begin{eqnarray}
     D(\rho\|\sigma) := \begin{cases}
\Tr(\rho(\log\rho-\log\sigma)), & \text{if } \operatorname{supp}(\rho)\subseteq \operatorname{supp}(\sigma), \\
\infty, & \text{otherwise}.
\end{cases}
\end{eqnarray}

        \item  For $\epsilon\in[0,1)$, the \textup{quantum hypothesis-testing relative entropy} \cite{Wang_2012} of $\rho$ with respect to $\sigma$ is defined as
        \begin{eqnarray}
            D^\epsilon_{H}(\rho\|\sigma):=-\log\min_{\substack{0\leq \Pi\leq \mathbb{I}\\\Tr[\Pi\rho]\geq1-\epsilon}}\Tr[\Pi\sigma].
        \end{eqnarray}
        Here $\{\Pi,\mathbb{I}-\Pi\}$ is a binary positive operator-valued measure (POVM).
    \end{itemize}  
\end{definition}

In the following,  we provide several important properties of these divergences that will be used repeatedly in the subsequent analysis.
\begin{enumerate}
    \item Given $\rho=\sum_ip_i\rho_i$ and $\sigma=\sum_ip_i\sigma_i$, where $\{p_i\}_i$ is a probability distribution, quantum relative entropy is jointly convex:
\begin{eqnarray}
D\!\left(\sum_i p_i\rho_i\,\middle\|\,\sum_i p_i\sigma_i\right)
\leq
\sum_i p_i D(\rho_i\|\sigma_i).
\end{eqnarray}

\item For product states of the form $\rho=\bigotimes_{i=1}^n\rho_i$ and $\sigma=\bigotimes_{i=1}^n\sigma_i$, we have the additivity property
\begin{equation}
D\!\left(
\bigotimes_{i=1}^n\rho_i
\,\middle\|\,
\bigotimes_{i=1}^n\sigma_i
\right)
=
\sum_{i=1}^nD(\rho_i\Vert\sigma_i).
\label{eq:add}
\end{equation}

\item  Both $D(\rho\|\sigma)$ and $D^\epsilon_{H}(\rho\|\sigma)$ satisfy the data processing inequality (DPI) \cite{Lindblad_1975}:
\begin{align}
    D(\mathcal E(\rho)\|\mathcal E(\sigma))\leq D(\rho\|\sigma)\quad\text{and}\quad D^{\epsilon}_H(\mathcal E(\rho)\|\mathcal E(\sigma))\leq D^{\epsilon}_H(\rho\|\sigma),
\end{align}
where $\mathcal E$ is an arbitrary CPTP map.

\item For normalised states $\rho$ and $\sigma$, by the quantum Pinsker's inequality~\cite{OhyaPetz1993}
\begin{eqnarray}
    D(\rho\Vert\sigma)\geq \frac{1}{2\ln 2}\|\rho-\sigma\|_1^2=\frac{2}{\ln 2}T(\rho,\sigma)^2,
\end{eqnarray}
Umegaki divergence relates to the trace distance between the states $T(\rho,\sigma):=\frac12\|\rho-\sigma\|_1$ , known as normalized $1$-norm distance.

\item The upper bound on the hypothesis testing relative entropy $D^\epsilon_H(\rho\|\sigma)$ is given by 
    \begin{eqnarray}
        D^\epsilon_H(\rho\|\sigma)\leq\frac{1}{1-\epsilon}[D(\rho\|\sigma)+h(\epsilon)],\quad 0\leq\epsilon<1,
        \label{eq:hr_bound}
    \end{eqnarray}
    where $h(\epsilon)$ is the binary entropy. This relation  is very well-known and can be proved using DPI of relative entropy as shown in Ref.~\cite{Wang_2012}.
\end{enumerate}

For any state divergence $\mathbf D$ in Def.~\ref{def:sd} and two channels $\mathcal N, \mathcal M\in \mathsf{CPTP}(A\to B)$,  one can define the corresponding unstabilized channel divergence as\footnote{$\mathbf D$ is used as a placeholder for $D^\epsilon_H$ or $D$. Similarly, $\mathbf d$ is used as a placeholder for $d^\epsilon_H$ or $d$.} 
\begin{eqnarray}
    {\mathbf d}(\mathcal N\|\mathcal M):=\sup_{\rho\in\mathfrak D(A)} \mathbf D(\mathcal N(\rho)\|\mathcal M(\rho)),
\end{eqnarray}
where the supremum is taken over all density operators $\rho$ on the input space $A$ of the channel $\mathcal N$. The stabilized channel divergence of $\mathcal N$ 
 with respect to $\mathcal M$ is defined as 
\begin{eqnarray}
    \widetilde{\mathbf d}(\mathcal N\|\mathcal M):=\sup_{\rho\in\mathfrak D(R\otimes A)} \mathbf D((\operatorname{id}_R\otimes \mathcal N)[\rho]\|(\operatorname{id}_R\otimes\mathcal M)[\rho]),
    \label{eq:stablizedd}
\end{eqnarray}
where by the virtue of purification and DPI, the supremum can be restricted to pure states such that $R$ is isomorphic to $A$ \cite{leditzky2018approaches,bergh2023composite}.

In the subsequent analysis, we will encounter quantities that depend simultaneously on $n$ number of channel uses and error tolerance $\epsilon$. However, our main point of interest will be for $\epsilon\to 0$ and asymptotic scenario with $n\to\infty$. In general, it may happen that a quantity of interest does not converge as $n\to\infty$ for any finite $\epsilon$, although, in many cases it is possible to show that the quantity converges after taking the additional limit $\epsilon\to 0$ outside. In accordance with Ref.~\cite{bergh2023composite}, we therefore adopt the following definition. 
\begin{definition}
    For a function $g: \mathbb N\times(0,1)\to\mathbb R$ and $a\in\overline{\mathbb R}$ where $\overline{\mathbb R}:=\mathbb R\cup\{-\infty,\infty\}$,
    we write
    \begin{eqnarray}
        \lim_{\epsilon\to0}\mathop{\overline{\underline{\lim}}}\limits_{n\to\infty}g(n,\epsilon)=a,
    \end{eqnarray}
    if
    \begin{eqnarray}
        \lim_{\epsilon\to0}\liminf_{n\to\infty}g(n,\epsilon)=a,\text{  and,  }\lim_{\epsilon\to0}\limsup_{n\to\infty}g(n,\epsilon)=a.
    \end{eqnarray}
\end{definition}

The following lemma will be useful in handling quantities involving optimization over input states and channels.
\begin{lemma}
    Let $\mathtt Y$ be a nonempty set and a function $f:\mathtt Y\cross \mathbb N\to \mathbb R$. Then $f(y,n)$ follows the inequalities
    \begin{eqnarray}
        \liminf_{n\to\infty}\sup_{y\in \mathtt Y} f(y,n)&\geq&\sup_{y\in \mathtt Y}\liminf_{n\to\infty}f(y,n),\label{eq:infsup}\\
    \text{and,   } \limsup_{n\to\infty}\sup_{y\in \mathtt Y} f(y,n)&\geq&\sup_{y\in \mathtt Y}\limsup_{n\to\infty}f(y,n).
    \label{eq:supsup}
    \end{eqnarray}
\end{lemma}
\begin{proof}
    For the above function $f(y,n)$, we can write
    \begin{eqnarray}
        &&\nonumber\sup_{y\in\mathtt Y}f(y,n)\geq f(y,n)~\forall n\in \mathbb N,y\in\mathtt Y,\\\nonumber
        \implies&&\liminf_{n\to\infty}\sup_{y\in\mathtt Y}f(y,n)\geq \liminf_{n\to\infty}f(y_1,n),~~\text{for any $y_1\in\mathtt Y$,}\\\nonumber\implies&&\liminf_{n\to\infty}\sup_{y\in\mathtt Y}f(y,n)\geq \sup_{y_1\in\mathtt Y}\liminf_{n\to\infty}f(y_1,n),\\\implies&&\liminf_{n\to\infty}\sup_{y\in\mathtt Y}f(y,n)\geq \sup_{y\in\mathtt Y}\liminf_{n\to\infty}f(y,n).
    \end{eqnarray}
    In a similar fashion Eq.~\eqref{eq:supsup} can also be proved.
\end{proof}
\subsection{Composite quantum state hypothesis testing}
Hypothesis testing deals with finding the truth between multiple possibilities from a given dataset. In a simple independent and identically distributed (i.i.d.) quantum state discrimination setting, given $n$ copies of a quantum state, promised to be either $\rho$ or $\sigma$, the task is to choose the right one between these two options. In a more general composite scenario, instead of a single state, one has to discriminate between two sets of states. Let us frame this case mathematically.

Suppose that we are given an unknown quantum state of a quantum system composed of $n$ copies of an elementary system, each with Hilbert space $A$, belonging to either   the set $\mathfrak R_n:=\{\rho_{(n)}:\rho_{(n)}\in\mathfrak D(A^{\otimes n})\}\subset\mathfrak{D}(A^{\otimes n})$ or the set $\mathfrak S_n:=\{\sigma_{(n)}:\sigma_{(n)}\in\mathfrak D(A^{\otimes n})\}\subset\mathfrak{D}(A^{\otimes n})$. More precisely, the null hypothesis $H_0$ corresponds to the case where the given state belongs to the set $\mathfrak{R}_n$, whereas the alternative hypothesis $H_1$ corresponds to the case where the given state belongs to the set $\mathfrak{S}_n$. We collect these two families of states in the sequences ${\mathfrak R}:=(\mathfrak R_n)_n$ and ${\mathfrak S}:=(\mathfrak S_n)_n$ which, by a slight abuse of notation, will also be referred to as the hypotheses. The hypothesis ${\mathfrak R}$ is \textit{composite} if the sets $\mathfrak R_n$ are not singleton sets. However, given a set of states $\mathfrak R\subset \mathfrak D(A)$, if the state on $n$-partite system is in the form of tensor product across copies, i.e., the family $\mathfrak R_n:=\{\otimes_{i=1}^n\rho_{i}|\rho_{1},\ldots,\rho_{n}\in\mathfrak R_1\}$ is called \textit{composite arbitrarily varying hypothesis} while the \textit{composite i.i.d. hypothesis} refers to $\mathfrak R_n:=\{\rho^{\otimes n}|\rho\in\mathfrak R_1\}$.

To address the discrimination problem, we perform a two-outcome POVM, which is fully defined by one of its elements $0\leq\Pi_n\leq \mathbb I$. We use the convention that the outcome $\Pi_n (\text{resp. } \mathbb I-\Pi_n)$ leads to the inference that the state belongs to the set $\mathfrak R_n(\text{resp. } \mathfrak S_n)$. In this scenario, we may encounter two distinct kinds of error, namely Type-I and Type-II, which are defined as
\begin{eqnarray}
    \nonumber \text{Type-I}\to\text{Probability}[\text{we infer the state is $\sigma_{ (n)}$}|\text{state is actually $\rho_{(n)}$}],\\\nonumber
    \text{Type-II}\to\text{Probability}[\text{we infer the state is $\rho_{ (n)}$}|\text{state is actually $\sigma_{(n)}$}].
\end{eqnarray}
Now, the error probabilities depend on the particular states we are getting from the families \(\mathfrak R_n\) and \(\mathfrak S_n\). This naturally leads to a worst-case formulation in which the performance of a test is evaluated against the most adverse choice of states within the respective hypothesis classes. Accordingly, for the POVM \(\Pi_n\), the Type-I and Type-II error probabilities are given by
\begin{eqnarray}
    \nonumber \alpha(\Pi_n,{\mathfrak R_n})&=&\sup_{\rho_n\in \mathfrak R_n}\Tr[(\mathbb I-\Pi_n)\rho_{(n)}],\\
    \text{and,}~~\beta(\Pi_n,{\mathfrak S_n})&=&\sup_{\sigma_n\in \mathfrak S_n}\Tr[\Pi_n\sigma_{(n)}],
\end{eqnarray}
respectively. For a given $n$, in the context of asymmetric hypothesis testing, the aim is to minimize the Type-II error over all possible POVMs by restricting the Type-I error below a certain threshold $\epsilon\in(0,1)$. We denote the resulting Type-II error by $\beta_\epsilon(\mathfrak R_n\|\mathfrak S_n)$ which can be written as 
\begin{eqnarray}
    \nonumber \beta_\epsilon(\mathfrak R_n\|\mathfrak S_n) &=& \inf_{0\leq \Pi_n\leq \mathbb I}\{\beta(\Pi_n,{\mathfrak S_n}):\sup_{\rho_{(n)}\in\mathfrak R_n}\Tr[(\mathbb I-\Pi_n)\rho_{(n)}]\leq \epsilon\} \\&=& \inf_{0\leq \Pi_n\leq \mathbb I}\{\sup_{\sigma_{(n)}\in\mathfrak S_n}\Tr[\Pi_n\sigma_{(n)}]:\sup_{\rho_{(n)}\in\mathfrak R_n}\Tr[(\mathbb I-\Pi_n)\rho_{(n)}]\leq \epsilon\}.
    \label{eq:betaep}
\end{eqnarray}
By Sion's minimax theorem \cite{Sion1958OnGM} and linearity of the error probabilities,  Eq.~\eqref{eq:betaep} can be re-expressed as (see Lemma~$31$ of Ref.~\cite{Fang2026gen}, and Ref.~\cite{Lami_2026_universalquantumresourcedistillation}) 
\begin{eqnarray}
    \nonumber \beta_\epsilon(\mathfrak R_n\|\mathfrak S_n) &:=& \sup_{\substack{\rho_{(n)}\in\operatorname{conv}(\mathfrak R_n)\\\sigma_{(n)}\in\operatorname{conv}(\mathfrak S_n)}}\inf_{0\leq \Pi_n\leq \mathbb I}\{\Tr[\Pi_n\sigma_{(n)}]:\Tr[(\mathbb I-\Pi_n)\rho_{(n)}]\leq \epsilon\}\\&=&\sup_{\substack{\rho_{(n)}\in\operatorname{conv}(\mathfrak R_n)\\\sigma_{(n)}\in\operatorname{conv}(\mathfrak S_n)}}\beta_\epsilon(\rho_{(n)}\|\sigma_{(n)}),
\end{eqnarray}
where $\beta_\epsilon(\rho_{(n)}\|\sigma_{(n)}):=\inf_{0\leq \Pi_n\leq \mathbb I}\{\Tr[\Pi_n\sigma_{(n)}]:\Tr[(\mathbb I-\Pi_n)\rho_{(n)}]\leq \epsilon\}$.
The quantum hypothesis testing relative entropy for the two sets $\mathfrak R_n$ and $\mathfrak S_n$ can be written as
\begin{eqnarray}
    D^\epsilon_H(\operatorname{conv}(\mathfrak R_n)\|\operatorname{conv}(\mathfrak S_n)):=-\log\beta_\epsilon(\mathfrak R_n\|\mathfrak S_n)=\inf_{\substack{\rho_{(n)}\in\operatorname{conv}(\mathfrak R_n)\\\sigma_{(n)}\in\operatorname{conv}(\mathfrak S_n)}}D^\epsilon_H(\rho_{(n)}\|\sigma_{(n)}).
\end{eqnarray}
Now the optimal exponent at which $\beta_\epsilon(\mathfrak R_n\|\mathfrak S_n)$ decays asymptotically with $n$, is characterized by \textit{Stein's exponent} as
\begin{eqnarray}
    \operatorname{Stein}({\mathfrak{R}}\|{\mathfrak{S}}):= E({\mathfrak{R}}\|{\mathfrak{S}}) = \lim_{\epsilon\to0^+}\liminf_{n\to\infty}\frac1nD^\epsilon_H(\operatorname{conv}(\mathfrak R_n)\|\operatorname{conv}(\mathfrak S_n)).
\end{eqnarray}
When both $H_0$ and $H_1$ are simple i.i.d. hypothesis, which means $\mathfrak R_n=\rho^{\otimes n}$ and $\mathfrak S_n=\sigma^{\otimes n}$ such that $\rho\in \mathfrak R_1\subset \mathfrak D(A), \sigma\in\mathfrak S_1\subset\mathfrak D(A)$, we recover one of the cornerstone results of quantum information theory, namely quantum Stein's Lemma \cite{HiaiPetz1991,OgawaNagaoka2000}:
\begin{eqnarray}
    \lim_{n\to\infty}\frac{1}{n}D^\epsilon_H(\rho^{\otimes n}\|\sigma^{\otimes n})=D(\rho\|\sigma)~\forall\epsilon\in(0,1),
\end{eqnarray}
and the corresponding Stein exponent as $\operatorname{Stein}(\rho\|\sigma)=D(\rho\|\sigma)$. 

\section{Composite i.i.d. quantum channel discrimination}\label{sec:compchannel}
The task of composite i.i.d. quantum channel discrimination is an extension of the composite i.i.d. state discrimination, where instead of states the objects are replaced by channels \cite{bergh2023composite}. More precisely, let $\mathcal E\in\mathsf{CPTP}(A\to B)$ be an unknown channel which is promised to belong to the one of the two bounded sets of channels, $\mathsf R_n^\mathsf{iid}:=\{\mathcal X^{\otimes n}:\mathcal X\in\mathsf R\}$ or $\mathsf S_n^\mathsf{iid}:=\{\mathcal Y^{\otimes n}:\mathcal Y\in\mathsf S\}$ where $\mathsf{R,S}\subset \mathsf{CPTP}(A\to B)$. Given $n$ uses of $\mathcal E$, the objective is to determine to which of these two families it belongs to.

In contrast to quantum state discrimination, multiple uses of a quantum channel admit a more complex class of discrimination strategies. The channel uses can, for example, be queried in \textit{parallel} or sequentially in a \textit{adaptive} manner with possibility of having entangled states, auxiliary quantum memories, intermediate control operations. Even when we restrict our attention to the parallel strategies, the optimization over input states adds another nontrivial degree of freedom along with finding the best POVM. In the following, let us briefly illustrate three particular classes of parallel strategies according to the correlation structure imposed on the probe states, that are relevant for this work.


\begin{figure}[t!]
\centering
\includegraphics[scale=0.2]{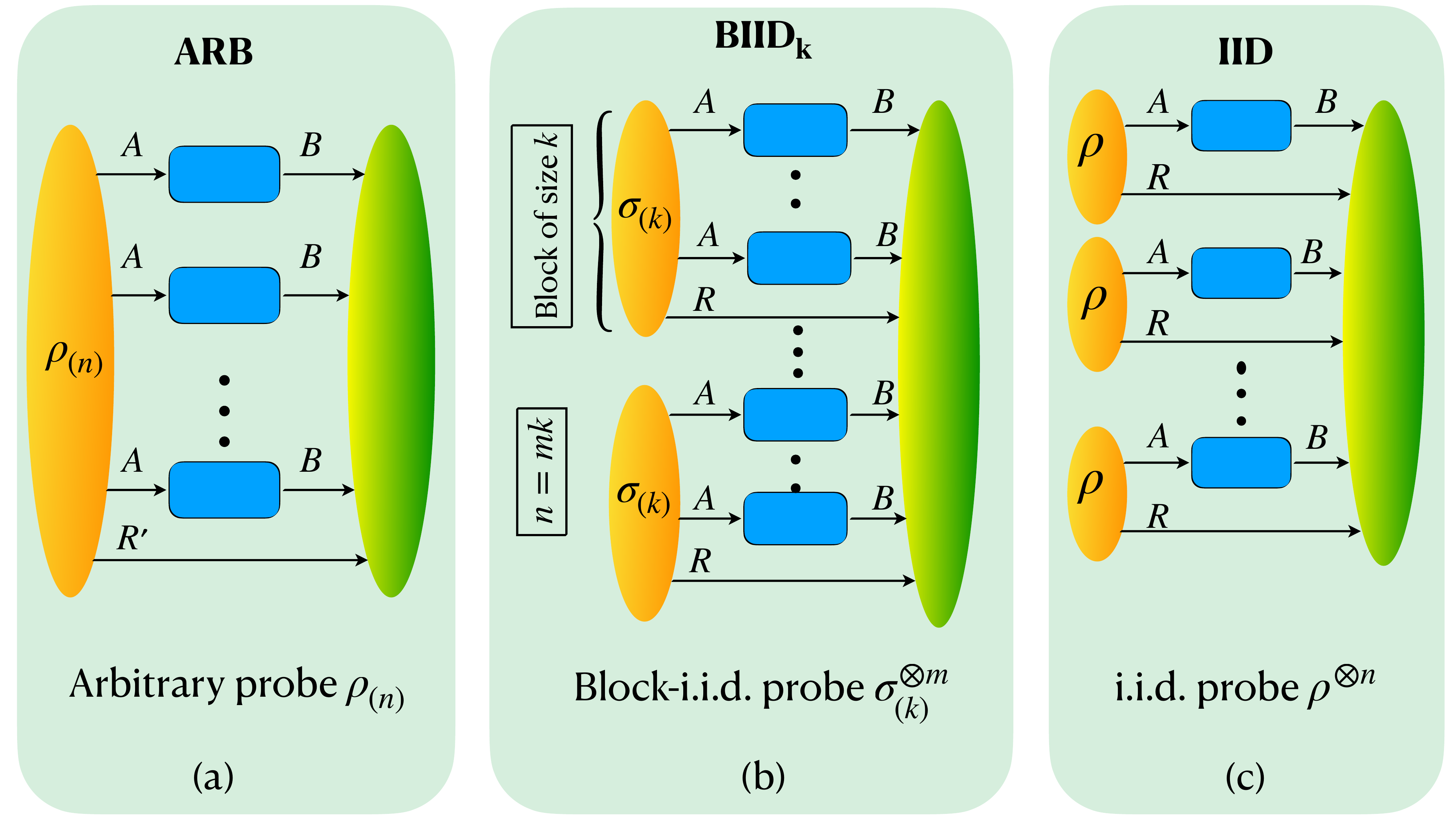}
\caption{Depiction of different strategies of channel discrimination under parallel usages. Each blue box denotes the unknown quantum channel $\mathcal E $, which belongs to either $\mathsf{R}$ (null hypothesis) or $\mathsf{S}$ (alternative hypothesis). Decoding is optimized over global measurements. (a) ARB strategy -- input probe states $\rho_{(n)}\in\mathfrak D(R'\otimes A^{\otimes n})$ have arbitrary correlation structure. (b) $\biid{k}$ strategy -- $n$ channels are divided into $m$ blocks; each of them having $k$ number of channels. Input state in each block, $\sigma_{(k)}\in \mathfrak D(R\otimes A^{\otimes k})$, can have arbitrary correlation but no inter-block correlation is allowed. (c) IID strategy -- input probes of each of the $n$ channels are prepared independently and identically in $\rho\in \mathfrak D(R\otimes A)$ .}
\label{fig:parallel}
\vspace{-.5cm}
\end{figure}

\textbf{Arbitrary-probe strategy}. We first consider the most general parallel probe allowed for $n$ uses of the channel as depicted in Fig.~\ref{fig:parallel}(a). Let $\rho_{(n)}\in \mathfrak D(R'\otimes A^{\otimes n})$ be an arbitrary joint state on $n$ channel inputs and an auxiliary memory system $R'$ . There are no restrictions on the correlations, classical or quantum, among the subsystems of $\rho_{(n)}$. Consequently, the output state after the channel action becomes  $(\mathrm{id}_{R'}\otimes\mathcal E^{\otimes n})[\rho_{(n)}]\in\mathfrak D(R'\otimes B^{\otimes n})$ where the auxiliary system $R'$ remains undisturbed. We refer to this class of parallel discrimination protocols employing such unrestricted probe states as  ARB strategy.


\textbf{Block-i.i.d. strategy}. Next, we introduce an intermediate class of strategies that interpolate between ARB and fully i.i.d. probe states (discussed subsequently). Suppose $n=mk$, where $k\leq n$ denotes the block size and $m$ is the number of blocks (see Fig.~\ref{fig:parallel}(b)). Within each block, $k$ channel inputs may be prepared in an arbitrary correlated state $\sigma_{(k)}\in\mathfrak D(R\otimes A^{\otimes k})$, including potential correlations with a reference system $R$ within each block. Now, $\sigma_{(k)}$ is prepared in an i.i.d. manner across the $m$ blocks of $k$ channels each, resulting in a global input state $\sigma_{(k)}^{\otimes m}\in\mathfrak D({R}^{\otimes m}\otimes A^{\otimes n})$. Therefore, arbitrary correlations are allowed within each block, while correlations between different blocks are excluded. Hence, the corresponding output state is $\left((\id_{R}\otimes \mathcal E^{\otimes k})[\sigma_{(k)}]\right)^{\otimes m}$. We call this class of protocols having block size $k$ as the $k$-block-i.i.d. ($\biid{k}$) strategy. Importantly, block size $k$ controls the structure of correlation available between the probes. Increasing $k$ enlarges the the set of admissible space of input states. In particular, for $k=n$, the block encompasses all $n$ uses of the channel; the BIID strategy reduces to the unrestricted ARB strategy, i.e., $\biid{n}\equiv\text{ARB}$.

Note, however, that for fixed $n$ and $k$, it may happen that $n$ is not a multiple of $k$. In that case, we can partition the $n$ channel uses into $\left\lfloor\frac{n}{k}\right\rfloor$ blocks of size $k$ and leave the remaining $1\leq (n\bmod k)\leq k-1$ channels unused. Since $\left(\frac{n\bmod k}{n}\right)\to 0$ as $n\to \infty$, these discarded channel uses do not affect the performance in asymptotic case.


\textbf{Fully i.i.d. strategy}. At the opposite extreme, one may adopt a fully i.i.d. probe $\rho\in \mathfrak D(R\otimes A)$ for every channel use, i.e., the global input state is $\rho^{\otimes n}$ (Fig.~\ref{fig:parallel}(c)). Therefore, the output state is $\left((\operatorname{id}_{R} \otimes \mathcal{E})[\rho]\right)^{\otimes n}$. This defines the IID strategy. Notice that $\biid{1}\equiv\text{IID}$ which means IID is obtained from block-i.i.d. strategy by choosing minimal block length $k=1$.

These three probe classes therefore form a natural hierarchy,
\begin{eqnarray}
    \text{IID}\subset\biid{k}\subset\text{ARB},~\forall\,k\in[2,n-1]; \text{ with } \biid{n}\equiv\text{ARB}, \text{ and } \biid{1}\equiv\text{IID},
\end{eqnarray}
corresponding to the sets of admissible parallel probe states. The block-i.i.d. construction is particularly useful as it can encompass both the complex ARB strategy and the simpler IID setting. By varying the  block size $k$, one can analyze how the correlation across an increasing number of channel uses affect the channel discrimination performance. In what follows, $\biid{k}$ will serve as the central framework for our analysis of a systematic investigation of non-mixed-untarity of a given channel within the framework of composite i.i.d. channel discrimination. Note that, in each of the strategies discussed above, setting the dimension of the auxiliary memory to one defines the corresponding \textit{memoryless} strategy. Other strategies, including adaptive ones are out of the scope of this paper and will not be discussed further. 


\section{Hypothesis Testing of quantum channels using block-iid input probes: General framework}
\label{sec:block-iid}
In this section, we consider the general problem of composite i.i.d. hypothesis testing of quantum channels under parallel strategies, with the admissible probe strategy being $\biid{k}$. 
As input probes, we consider states of the form
\begin{align}
    \rho_{(n)}=\sigma_{(k)}^{\otimes m},\qquad n=mk,    
\end{align}
where 
\begin{itemize}
    \item for memory-assisted strategies, $\rho_{(n)}\in\mathfrak{D}\left(R'\otimes A^{\otimes n}\right)$ and $\sigma_{(k)}\in\mathfrak{D}\left(R\otimes A^{\otimes k}\right)$ with $R'={R}^{\otimes m}$, and
    \item for memoryless strategies, $\rho_{(n)}\in\mathfrak{D}\left(A^{\otimes n}\right)$ and $\sigma_{(k)}\in\mathfrak{D}\left(A^{\otimes k}\right)$.
\end{itemize}


Aligning with our primary objective of testing whether a given channel is a specified non-mixed-unitary channel or belongs to the set of mixed-unitary channels, we consider the more general problem of distinguishing a specified channel lying outside a closed convex set of channels from that set. Accordingly, our objective is to distinguish between the following two hypotheses regarding an unknown channel \(\mathcal{E}\):
    \begin{align}
        H_0:\quad &\mathcal{E}\in\mathsf{R},\nonumber\\
        H_1:\quad &\mathcal{E}\in\mathsf{S}=\{\mathcal{N}\},\nonumber
    \end{align}
where $\mathsf{R}$ is an arbitrary closed convex set of quantum channels and $\mathsf S$ is a singleton set of a fixed channel $\mathcal N\notin \mathsf R$. With $n$ channel uses, these induce the following composite and simple i.i.d. hypotheses
    \begin{align}
        \mathsf{R}^{\mathrm{iid}}_n:=\{\Phi^{\otimes n}:\Phi\in\mathsf{R}\},~~\text{and, }~~
        \mathsf{S}^{\mathrm{iid}}_n:=\{\mathcal{N}^{\otimes n}\}\equiv\mathcal{N}^{\mathrm{iid}}_n,
    \end{align}
    respectively. Similarly to the  composite quantum state hypothesis testing, here we collect the sequences of sets of channels as $\mathsf R^\mathrm{iid}=(\mathsf R^\mathrm{iid}_n)_n$ and $\mathcal N^\mathrm{iid}=( \mathcal N^\mathrm{iid}_n)_n$ and refer them as composite hypothesis as well.


\subsection{Stein Exponents}
We now turn to characterizing the quantum Stein's exponents for the aforementioned channel discrimination problem under different strategies. To this end, we refer to Stein's exponent in ARB strategy with and without memory as 
\begin{eqnarray}
    \nonumber \widetilde E(\mathsf R^{\mathrm{iid}}\|\mathcal N^\mathrm{iid})&=&\lim_{\epsilon\to0}\mathop{\overline{\underline{\lim}}}\limits_{n\to\infty}\frac{1}{n}\widetilde d^\epsilon_{H}(\{\Phi^{\otimes n}\}\|\mathcal N^{\otimes n}),\\\text{and,}~~
     E(\mathsf R^{\mathrm{iid}}\|\mathcal N^\mathrm{iid})&=&\lim_{\epsilon\to0}\mathop{\overline{\underline{\lim}}}\limits_{n\to\infty}\frac{1}{n} d^\epsilon_{H}(\{\Phi^{\otimes n}\}\|\mathcal N^{\otimes n}),
\end{eqnarray}
respectively, with $\Phi\in\mathsf{R}$; whereas in $\biid{k}$ strategy, the Stein's exponents with and without memory are denoted by $\widetilde{E}^k\left(\mathsf{R}^{\mathrm{iid}}\middle\|\mathcal{N}^{\mathrm{iid}}\right)$ and ${E}^k\left(\mathsf{R}^{\mathrm{iid}}\middle\|\mathcal{N}^{\mathrm{iid}}\right)$, respectively. Mathematically, we can write
\begin{eqnarray}
    \nonumber \widetilde E^k(\mathsf R^{\mathrm{iid}}\|\mathcal N^\mathrm{iid})=\lim_{\epsilon\to0}\mathop{\overline{\underline{\lim}}}\limits_{m\to\infty}\frac{1}{mk}\sup_{\sigma_{(k)}\in\mathfrak D(R\otimes A^{\otimes k})}D^{\epsilon}_{H}\big(\{(\mathrm{id}_R\otimes\Phi^{\otimes k}[\sigma_{(k)}])^{\otimes m}\}\|[(\mathrm{id}_R\otimes\mathcal N^{\otimes k}[\sigma_{(k)}])^{\otimes m}\big),
\end{eqnarray}
where $\sigma_{(k)}\in\mathfrak D(R\otimes A^{\otimes k})$ with $R$ being isomorphic to $A^{\otimes k}$ according to the definition of stabilized channel divergence in Eq.~\eqref{eq:stablizedd}, i.e., $R\cong A^{\otimes k}$. Moreover, let us define the output states as $\sigma_{{(k)},\Phi}:=(\id_R\otimes\Phi^{\otimes k})[\sigma_{(k)}]\in\mathfrak D(R\otimes B^{\otimes k})$ and $\sigma_{{(k)},\mathcal N}:=(\id_R\otimes\mathcal N^{\otimes k})[\sigma_{(k)}]\in\mathfrak D(R\otimes B^{\otimes k})$. From now on, we will write $\mathfrak D(C\otimes D)$ alternatively as $\mathfrak D(CD)$ by omitting the tensor-product sign for brevity. 
The Stein exponent without memory for $\biid{k}$ strategy can be expressed as,
\begin{eqnarray}
    \nonumber  E^k(\mathsf R^{\mathrm{iid}}\|\mathcal N^\mathrm{iid})=\lim_{\epsilon\to0}\mathop{\overline{\underline{\lim}}}\limits_{m\to\infty}\frac{1}{mk}\sup_{\sigma_{(k)}\in\mathfrak D(A^{\otimes k})}D^{\epsilon}_{H}\big(\{(\Phi^{\otimes k}[\sigma_{(k)}])^{\otimes m}\}\|(\mathcal N^{\otimes k}[\sigma_{(k)}])^{\otimes m}\big),
\end{eqnarray}
by choosing the dimension of $R$ as one. The usual IID strategy is recovered by choosing $k=1$ where the Stein exponents are $\widetilde E^1(\mathsf R^{\mathrm{iid}}\|\mathcal N^\mathrm{iid})$ and $ E^1(\mathsf R^{\mathrm{iid}}\|\mathcal N^\mathrm{iid})$.

Let us now present one of the main results characterizing the quantum Stein exponents for the channel discrimination problem under the \(\biid{k}\) strategy.

\begin{theorem}[Stein exponents for $\biid{k}$ strategy]\label{theo:stein_gen}
    Let the sequence $\mathsf{R}^{\mathrm{iid}}:=\left(\mathsf{R}^{\mathrm{iid}}_n\right)_n$ be a composite i.i.d. null channel hypothesis and the sequence $\mathcal{N}^{\mathrm{iid}}:=\left(\mathcal{N}^{\otimes n}\right)_n$ be a simple i.i.d. alternative channel hypothesis, with $\mathsf R$ being a closed-convex set of quantum channels and $\mathcal{N}\notin\mathsf{R}$. Under parallel channel usages with the $\biid{k}$ strategy, for a fixed block size $k$, the quantum Stein exponents with and without auxiliary memory are given respectively by
    \begin{align}
        &\widetilde{E}^k\left(\mathsf{R}^{\mathrm{iid}}\middle\|\mathcal{N}^{\mathrm{iid}}\right)=\sup_{\sigma_{(k)}\in\mathfrak{D}(R\otimes A^{\otimes k})}\min_{\Phi\in\mathsf{R}} \frac{1}{k}\,D\left((\mathrm{id}_R\otimes \Phi^{\otimes k})[\sigma_{(k)}]\middle\|(\mathrm{id}_R\otimes\mathcal{N}^{\otimes k})[\sigma_{(k)}]\right),\label{gen_stein_mem}\\
        &E^k\left(\mathsf{R}^{\mathrm{iid}}\middle\|\mathcal{N}^{\mathrm{iid}}\right)=\sup_{\sigma_{(k)}\in\mathfrak{D}(A^{\otimes k})}\min_{\Phi\in\mathsf{R}} \frac{1}{k}\,D\left(\Phi^{\otimes k}[\sigma_{(k)}]\middle\|\mathcal{N}^{\otimes k}[\sigma_{(k)}]\right)\label{gen_stein},
    \end{align}
    where it suffices to choose $R\cong A^{\otimes k}$ for all $k$.
\end{theorem}
\begin{proof}
    Given total $n$ number of channels, in the block-i.i.d. scenario, we first characterize the Stein exponent with auxiliary memory.
    
\noindent\textbf{Achievability.}  For the achievable part, we have 
\begin{eqnarray}
    \nonumber\widetilde E^k(\mathsf R^{\mathrm{iid}}\|\mathcal N^\mathrm{iid})&\geq& \lim_{\epsilon\to0}\liminf_{m\to\infty}\sup_{\sigma_{(k)}\in\mathfrak D(R A^{\otimes k})}\frac{1}{mk}D^{\epsilon}_{H}\big(\{\sigma_{{(k)},\Phi}^{\otimes m}\}\|\sigma_{{(k)},\mathcal N}^{\otimes m}\big),\\\nonumber
    &\overset{\eqref{eq:infsup}}{\geq}& \lim_{\epsilon\to0}\sup_{\sigma_{(k)}\in\mathfrak D(R A^{\otimes k})}\frac{1}{k}\left[\liminf_{m\to\infty}\frac{1}{m}D^{\epsilon}_{H}\big(\{\sigma_{{(k)},\Phi}^{\otimes m}\}\|\sigma_{{(k)},\mathcal N}^{\otimes m}\big)\right],\\\nonumber&=&\lim_{\epsilon\to0}\sup_{\sigma_{(k)}\in\mathfrak D(R A^{\otimes k})}\inf_{\sigma_{{(k)},\Phi}}\frac 1kD(\sigma_{{(k)},\Phi}\|\sigma_{{(k)},\mathcal N}).    \\&=&\sup_{\sigma_{(k)}\in\mathfrak D(R A^{\otimes k})}\inf_{\Phi\in\mathsf R}\frac 1kD\left((\id_R\otimes\Phi^{\otimes k})[\sigma_{(k)}]\|(\id_R\otimes\mathcal N^{\otimes k})[\sigma_{(k)}]\right),
\end{eqnarray}
where in the third equality, following Eq.~$(7)$ of Ref.~\cite{Berta2021on}, we have used the fact that for a composite i.i.d. null state hypothesis $\{\rho^{\otimes m}:\rho\in \mathfrak R_1\subset \mathfrak D(A)\}$ and a simple i.i.d. alternate state hypothesis $\eta^{\otimes m}$ such that $\eta\notin \mathfrak R_1$ it can be shown that
\begin{eqnarray}
    \lim_{m\to\infty}\frac 1m D^\epsilon_H(\{\rho^{\otimes m}\}\|\eta^{\otimes m}) = \inf_{\rho\in \mathfrak R_1}\lim_{m\to\infty}\frac{1}{m}D^\epsilon_H(\rho^{\otimes m}\|\eta^{\otimes m})=\inf_{\rho\in \mathfrak R_1} D(\rho\|\eta).
\end{eqnarray}\\

\noindent\textbf{Converse.} To derive the converse part, we start with
\begin{eqnarray}
    \nonumber\widetilde E^k(\mathsf R^{\mathrm{iid}}\|\mathcal N^\mathrm{iid})&\leq& \lim_{\epsilon\to0}\limsup_{m\to\infty}\sup_{\sigma_{(k)}\in\mathfrak D(R A^{\otimes k})}\frac{1}{mk}D^{\epsilon}_{H}\big(\{\sigma_{{(k)},\Phi}^{\otimes m}\}\|\sigma_{{(k)},\mathcal N}^{\otimes m}\big),\\\nonumber&\leq&\lim_{\epsilon\to0}\limsup_{m\to\infty}\sup_{\sigma_{(k)}\in\mathfrak D(R A^{\otimes k})}\frac{1}{mk}\inf_{\Phi\in\mathsf R}D^{\epsilon}_{H}\big(\sigma_{{(k)},\Phi}^{\otimes m}\|\sigma_{{(k)},\mathcal N}^{\otimes m}\big),\\\nonumber&\overset{\eqref{eq:hr_bound}}\leq&\lim_{\epsilon\to0}\limsup_{m\to\infty}\sup_{\sigma_{(k)}\in\mathfrak D(R A^{\otimes k})}\frac{1}{mk}\inf_{\Phi\in\mathsf R}
    \frac{1}{1-\epsilon}\left[D\big(\sigma_{{(k)},\Phi}^{\otimes m}\|\sigma_{{(k)},\mathcal N}^{\otimes m}\big)+h(\epsilon)\right],\\\nonumber&=&\lim_{\epsilon\to0}\limsup_{m\to\infty}\sup_{\sigma_{(k)}\in\mathfrak D(R A^{\otimes k})}\inf_{\Phi\in\mathsf R}\frac{1}{k}
    \frac{1}{1-\epsilon}\left[D\big(\sigma_{{(k)},\Phi}\|\sigma_{{(k)},\mathcal N}\big)+\frac{h(\epsilon)}{m}\right],\\
    &=&\sup_{\sigma_{(k)}\in\mathfrak D(R A^{\otimes k})}\inf_{\Phi\in\mathsf R}\frac 1kD\left((\id_R\otimes\Phi^{\otimes k})[\sigma_{(k)}]\|(\id_R\otimes\mathcal N^{\otimes k})[\sigma_{(k)}]\right),
\end{eqnarray}
where the second inequality follows from the logic that discriminating a state $\eta^{\otimes m}$ from a set of $m$-copies of states $\{\rho^{\otimes m}\}$ can be done no better than discriminating $\eta^{\otimes m}$ from each individual state $\rho^{\otimes m}$ in the set, including the worst-case scenario, i.e., $D^\epsilon_H(\{\rho^{\otimes m}\}\|\sigma^{\otimes m})\leq\min_{\rho}D^\epsilon_H(\rho^{\otimes m}\|\sigma^{\otimes m})$ $\forall m$. The fourth equality follows from the additivity property (Eq.~\eqref{eq:add}) of Umegaki relative entropy under tensor product.

Furthermore, $\mathsf R$ is closed and bounded; hence a compact set. Since, Umegaki relative entropy is a lower semi-continuous function (see Exercise $7.22$ of Ref.~\cite{Holevo2012} and Lemma~$24$ of Ref.~\cite{bergh2023composite}) over a compact set of quantum channels for every input state $\sigma_{(k)}$, inner infimum is always achieved. Hence, we have 
\begin{eqnarray}
    \widetilde E^k(\mathsf R^{\mathrm{iid}}\|\mathcal N^\mathrm{iid}) = \sup_{\sigma_{(k)}\in\mathfrak D(R A^{\otimes k})}\min_{\Phi\in\mathsf R}\frac 1kD\left((\id_R\otimes\Phi^{\otimes k})[\sigma_{(k)}]\|(\id_R\otimes\mathcal N^{\otimes k})[\sigma_{(k)}]\right),
\end{eqnarray}
whereas by choosing $\dim(R)=1$, we can easily arrive at 
\begin{eqnarray}
    E^k(\mathsf R^{\mathrm{iid}}\|\mathcal N^\mathrm{iid}) = \sup_{\sigma_{(k)}\in\mathfrak D(A^{\otimes k})}\min_{\Phi\in\mathsf R}\frac 1kD\left(\Phi^{\otimes k}[\sigma_{(k)}]\|\mathcal N^{\otimes k}[\sigma_{(k)}]\right).
\end{eqnarray}
This completes the proof.
\end{proof}

The setting of block-i.i.d. probe states naturally defines a hierarchy of discrimination settings parameterized by the block size $k$, bridging two extremes. At one extreme, setting $k=1$ corresponds to fully i.i.d. probe states, i.e., states of the form $\rho^{\otimes n}$. At the other extreme, taking $k=n$ and the limit $k\to\infty$ removes the block-i.i.d. restriction altogether, thereby recovering the setting of arbitrary probe states. Consequently, Theorem~\ref{theo:stein_gen} immediately yields the Stein exponents for these two extreme cases.

\begin{corollary}[Stein exponents for fully i.i.d. probe states]
    The quantum Stein exponents of the above composite channel hypothesis testing problem under fully i.i.d. probe states are given by
    \begin{align}
        &\widetilde{E}^1\left(\mathsf{R}^{\mathrm{iid}}\middle\|\mathcal{N}^{\mathrm{iid}}\right)=\sup_{\rho\in\mathfrak{D}(R\otimes A)}\min_{\Phi\in\mathsf{R}} \,D\left((\mathrm{id}_R\otimes \Phi)[\rho]\middle\|(\mathrm{id}_R\otimes\mathcal{N})[\rho]\right),\label{eq:stein_mem_fully_iid}\\
        &E^1\left(\mathsf{R}^{\mathrm{iid}}\middle\|\mathcal{N}^{\mathrm{iid}}\right)=\sup_{\rho\in\mathfrak{D}(A)}\min_{\Phi\in\mathsf{R}} \,D\left(\Phi[\rho]\middle\|\mathcal{N}[\rho]\right),
    \end{align}
    where it suffices to choose $R\cong A$.
\end{corollary}

\begin{corollary}[Stein exponents for arbitrary probe states]\label{cor:arbitrary_probe}
    The quantum Stein exponents of the above composite channel hypothesis testing problem under arbitrary probe states are given by
    \begin{align}
        \label{eq:stein_arb1}&\widetilde{E}\left(\mathsf{R}^{\mathrm{iid}}\middle\|\mathcal{N}^{\mathrm{iid}}\right)=\lim_{k\to\infty}\,\sup_{\sigma_{(k)}\in\mathfrak{D}(R\otimes A^{\otimes k})}\min_{\Phi\in\mathsf{R}} \frac{1}{k}\,D\left((\mathrm{id}_R\otimes \Phi^{\otimes k})[\sigma_{(k)}]\middle\|(\mathrm{id}_R\otimes\mathcal{N}^{\otimes k})[\sigma_{(k)}]\right),\\
        \label{eq:stein_arb2}&E\left(\mathsf{R}^{\mathrm{iid}}\middle\|\mathcal{N}^{\mathrm{iid}}\right)=\lim_{k\to\infty}\,\sup_{\sigma_{(k)}\in\mathfrak{D}(A^{\otimes k})}\min_{\Phi\in\mathsf{R}} \frac{1}{k}\,D\left(\Phi^{\otimes k}[\sigma_{(k)}]\middle\|\mathcal{N}^{\otimes k}[\sigma_{(k)}]\right).    \end{align}
\end{corollary}
Note that the limits appearing in Eqs.~\eqref{eq:stein_arb1} and~\eqref{eq:stein_arb2} always exist. Specifically, the additivity of the quantum relative entropy under tensor products ensures that the sequences $(kE^k)_k$ and $(k\widetilde{E}^k)_k$ are super-additive in $k$. By Fekete's lemma~\cite{Fekete1923}, each limit exists and coincides with the supremum over all $k \ge 1$.
\begin{remark}
    Corollary~\ref{cor:arbitrary_probe} also follows as a special case of Theorem~10 in Ref.~\cite{bergh2023composite}. It is important to note that, without the block-i.i.d. restriction, the corresponding Stein exponents for arbitrary probe states are generally given by regularized expressions. In contrast, fully i.i.d. probe states lead to finite-letter characterizations of the Stein exponents, but impose a stringent restriction on the state space of admissible input probes. The block-i.i.d. setting provides a natural middle ground between these two extremes: for every finite block size $k$, Theorem~\ref{theo:stein_gen} gives a finite-letter characterization of the corresponding Stein exponents, while increasing $k$ progressively increases the number of channel inputs that can be correlated within a block.
\end{remark}

\subsection{Probe states with correlation constraints}\label{probeconstraints}

A second type of hierarchy emerges by constraining the correlations among the $k$ channel inputs within each block. Namely, the most general form of quantum correlation that a probe state across the $k$ channel inputs (and the auxiliary system, in the case of memory-assisted strategies) may possess is genuine multipartite entanglement. On the other hand, one may exclude entanglement across the \(k\) channel uses by restricting the probe states to be fully separable, i.e., of the form
\begin{align}
    \sigma_{(k)}=\sum_j p_j\left(\bigotimes_{i=1}^{k}\tau_{i}^{j}\right),\label{FS}
\end{align}
where
\begin{itemize}
    \item for memory-assisted strategies, $\sigma_{(k)}\in\mathfrak{D}(R\otimes A^{\otimes k})$ and $\tau_{i}^j\in\mathfrak{D}(\bar R\otimes A)$ for every $i$ and each $j$, 
    with $R=\bar R^{\otimes k}$, 
    and
    \item for memoryless strategies, $\sigma_{(k)}\in\mathfrak{D}(A^{\otimes k})$ and $\tau_i^{j}\in\mathfrak{D}(A)$ for every $i$ and each $j$.
\end{itemize}
Note that in the memory-assisted setting, while entanglement across different channel uses is strictly restricted, entanglement between each channel input and its corresponding auxiliary reference $\bar R$ is not necessarily restricted. These constraints are particularly relevant in practical scenarios where the preparation of highly entangled probe states is experimentally or operationally limited. We denote the quantum Stein exponents corresponding to fully separable probe states within each block by $\widetilde{E}^k_{\mathrm{Sep}}$ and $E^k_{\mathrm{Sep}}$. They are given by Eqs.~\eqref{gen_stein_mem} and \eqref{gen_stein} with the supremum now being restricted to the set of separable states of the form in Eq.~\eqref{FS}.


One may further consider the case where there are no correlations at all, not even classical, among the channel inputs. In this case, the probe state is restricted to the product form
\begin{align}
    \sigma_{(k)}=\bigotimes_{i=1}^k\tau_{i}.
\end{align}
Using the additivity of the Umegaki relative entropy (Eq.~\eqref{eq:add}), Theorem~\ref{theo:stein_gen} immediately yields the following characterization of the quantum Stein exponents. 

\begin{corollary}[Stein exponents for block-i.i.d. product probe states]\label{cor:stein_pro}
    The quantum Stein exponents of the above composite channel hypothesis testing problem under block-i.i.d. product probe states are given by
    \begin{align}
        &\widetilde{E}^k_{\mathrm{Pro}}\left(\mathsf{R}^{\mathrm{iid}}\middle\|\mathcal{N}^{\mathrm{iid}}\right)=\sup_{\tau_1,\cdots,\tau_k\in \mathfrak{D}(R\otimes A)}\min_{\Phi\in\mathsf{R}} \frac{1}{k}\sum_{i=1}^k\,D\left((\mathrm{id}_R\otimes \Phi)[\tau_i]\middle\|(\mathrm{id}_R\otimes\mathcal{N})[\tau_i]\right),\\
        &E^k_{\mathrm{Pro}}\left(\mathsf{R}^{\mathrm{iid}}\middle\|\mathcal{N}^{\mathrm{iid}}\right)=\sup_{\tau_1,\cdots,\tau_k\in\mathfrak{D}( A)}\min_{\Phi\in\mathsf{R}} \frac{1}{k}\sum_{i=1}^k D\left(\Phi[\tau_i]\middle\|\mathcal{N}[\tau_i]\right),
    \end{align}
    where $R\cong A$.
\end{corollary}

Similar correlation constraints can also be imposed on arbitrary probe states without the block-i.i.d. restriction. The corresponding Stein exponents can then be obtained by taking the appropriate regularization over the block size $k$.

The following lemma, which follows directly from Theorem~\ref{theo:stein_gen} and Corollary~\ref{cor:stein_pro}, characterizes when the corresponding Stein exponents vanish, and will be useful in establishing our later results.

\begin{lemma}\label{usefulcond}
For any $\mathcal{N}\notin\mathsf{R}$ and any fixed block size $k$, the quantum Stein exponents of the above channel hypothesis testing problem under block-i.i.d. probe states are zero if and only if 
\begin{itemize}
    \item in the memory-assisted strategy, for every
    $\sigma_{(k)}\in\mathfrak{D}(R\otimes A^{\otimes k})$ with
    $R\cong A^{\otimes k}$, there exists a channel
    $\Phi\in\mathsf{R}$ such that
    \begin{align}
        \left(\mathrm{id}_R\otimes\Phi^{\otimes k}\right)[\sigma_{(k)}] 
        =
        \left(\mathrm{id}_R\otimes\mathcal{N}^{\otimes k}\right)[\sigma_{(k)}] ;
    \end{align}

    \item in the memoryless strategy, for every
    $\sigma_{(k)} \in\mathfrak{D}(A^{\otimes k})$, there exists a
    channel $\Phi\in\mathsf{R}$ such that
    \begin{align}
        \Phi^{\otimes k}[\sigma_{(k)}]
        =
        \mathcal{N}^{\otimes k}[\sigma_{(k)}].
    \end{align}
\end{itemize}

In particular, for block-i.i.d. product probe states, the corresponding
quantum Stein exponents are zero if and only if
\begin{itemize}
    \item in the memory-assisted strategy, for every collection of
    $k$ input probes $\{\tau_1,\ldots,\tau_k\}\subset
    \mathfrak{D}(R\otimes A)$ with $R\cong A$, there
    exists a channel $\Phi\in\mathsf{R}$ such that
    \begin{align}
        \left(\mathrm{id}_R\otimes\Phi\right)[\tau_i]
        =
        \left(\mathrm{id}_R\otimes\mathcal{N}\right)[\tau_i],
        \qquad \forall i=1,\ldots,k;
    \end{align}

    \item in the memoryless strategy, for every collection of 
    $k$ input probes $\{\tau_1,\ldots,\tau_k\}\subset
    \mathfrak{D}(A)$, there exists a channel
    $\Phi\in\mathsf{R}$ such that
    \begin{align}
        \Phi[\tau_i]
        =
        \mathcal{N}[\tau_i],
        \qquad \forall i=1,\ldots,k.
    \end{align}
\end{itemize}
\end{lemma}


\subsection{Ordering relations}\label{subsec:order}

The different hierarchies discussed above induce corresponding ordering relations among the Stein exponents of the respective discrimination settings.

\begin{itemize}
    \item \textbf{Hierarchy between memoryless and memory-assisted strategies:}
    Since every memoryless discrimination strategy can be regarded as a memory-assisted strategy with a trivial (one-dimensional) reference system, the class of memory-assisted strategies strictly contains the class of memoryless strategies. Consequently, for any block size $k$, the corresponding Stein exponents satisfy
    \begin{align}
        E^k\left(\mathsf{R}^{\mathrm{iid}}\middle\|\mathcal{N}^{\mathrm{iid}}\right)
        \leq \widetilde{E}^k\left(\mathsf{R}^{\mathrm{iid}}\middle\|\mathcal{N}^{\mathrm{iid}}\right).
    \end{align}
    Similar ordering relations also hold when the probe states within each block are subject to any particular correlation constraint discussed in Sec.~\ref{probeconstraints}.

    \item \textbf{Hierarchy between strategies with different block sizes:} Increasing the block size enlarges the set of admissible probe states within each block. Any block-i.i.d. strategy with block size $k$ naturally includes fully i.i.d. probe states as a special case, obtained by choosing
    \begin{align*}
        \sigma_{(k)}=\rho^{\otimes k},
    \end{align*}
    where $\sigma_{(k)}$ is the input state of each block and $\rho$ is the input state for each channel use in a fully i.i.d. probe strategy. Therefore, optimizing over block-i.i.d. inputs cannot give a smaller value than optimizing over fully i.i.d. inputs. On the other hand, the set of arbitrary input probes (without the block-i.i.d. restriction) in the asymptotic channel-use limit contains block-i.i.d. input probes of any block size $k$ as a special case. Therefore, optimization over arbitrary probe states cannot give a smaller value than optimization over block-i.i.d. inputs. Hence, for every fixed block size $k\geq2$, we have
    \begin{align}
        &\widetilde{E}^1\left(\mathsf{R}^{\mathrm{iid}}\middle\|\mathcal{N}^{\mathrm{iid}}\right)\leq\widetilde{E}^k\left(\mathsf{R}^{\mathrm{iid}}\middle\|\mathcal{N}^{\mathrm{iid}}\right)\leq\widetilde{E}\left(\mathsf{R}^{\mathrm{iid}}\middle\|\mathcal{N}^{\mathrm{iid}}\right),\\
        \text{and}\quad&E^1\left(\mathsf{R}^{\mathrm{iid}}\middle\|\mathcal{N}^{\mathrm{iid}}\right)\leq E^k\left(\mathsf{R}^{\mathrm{iid}}\middle\|\mathcal{N}^{\mathrm{iid}}\right)\leq E\left(\mathsf{R}^{\mathrm{iid}}\middle\|\mathcal{N}^{\mathrm{iid}}\right).
    \end{align}
    
    Furthermore, consider two strategies $\biid{k}$ and $\biid{k'}$ with different block sizes $k$ and $k'$, respectively. Whenever $k'=\ell k$ for some $\ell\in\mathbb{N}$, every block-i.i.d. probe state $\sigma_{(k)}$ of block size $k$ can be used to construct a block-i.i.d. probe state of block size $k'$ by choosing
    \begin{align*}
    \sigma_{(k')}=\sigma_{(k)}^{\otimes \ell}.
    \end{align*}
    Therefore, the set of input probes in a block-i.i.d. strategy with block size $k'$ contains the set of input probes in a block-i.i.d. strategy with block size $k$. Hence, for every block size $k$,
    \begin{align}
    \left.
    \begin{aligned}
        \widetilde{E}^k\left(\mathsf{R}^{\mathrm{iid}}\middle\|\mathcal{N}^{\mathrm{iid}}\right)&\leq \widetilde{E}^{\ell k}\left(\mathsf{R}^{\mathrm{iid}}\middle\|\mathcal{N}^{\mathrm{iid}}\right),\\
        E^k\left(\mathsf{R}^{\mathrm{iid}}\middle\|\mathcal{N}^{\mathrm{iid}}\right)&\leq E^{\ell k}\left(\mathsf{R}^{\mathrm{iid}}\middle\|\mathcal{N}^{\mathrm{iid}}\right)
        \end{aligned}
        \right\}\qquad\forall \ell\in\mathbb{N}.
    \end{align}

    \item{\textbf{Hierarchy between strategies with different correlation constraints:}} For any fixed block size $k\geq2$, the set of block-i.i.d. product probe states is a strict subset of the set of block-i.i.d. fully separable probe states, which, in turn, is a strict subset of the set of block-i.i.d. probe states without any restriction on the correlations. Consequently, the corresponding Stein exponents satisfy
    \begin{align}
        &\widetilde{E}^k_{\mathrm{Pro}}\left(\mathsf{R}^{\mathrm{iid}}\middle\|\mathcal{N}^{\mathrm{iid}}\right)\leq \widetilde{E}^k_{\mathrm{Sep}}\left(\mathsf{R}^{\mathrm{iid}}\middle\|\mathcal{N}^{\mathrm{iid}}\right)\leq \widetilde{E}^k\left(\mathsf{R}^{\mathrm{iid}}\middle\|\mathcal{N}^{\mathrm{iid}}\right),\\
        &E^k_{\mathrm{Pro}}\left(\mathsf{R}^{\mathrm{iid}}\middle\|\mathcal{N}^{\mathrm{iid}}\right)\leq E^k_{\mathrm{Sep}}\left(\mathsf{R}^{\mathrm{iid}}\middle\|\mathcal{N}^{\mathrm{iid}}\right)\leq E^k\left(\mathsf{R}^{\mathrm{iid}}\middle\|\mathcal{N}^{\mathrm{iid}}\right).
    \end{align}
        
\end{itemize}

\section{Hypothesis Testing against mixed-unitary channels: Stein exponents as resource monotones}
\label{sec:MU-hypothesis}
We now specialize our general framework to the case where the null hypothesis $H_0$ is given by the set of mixed-unitary channels, i.e., $\mathsf{R}=\mathsf{MU}$, and investigate the resulting hierarchy of quantum Stein exponents for distinguishing a fixed non-mixed-unitary channel from mixed-unitary channels.

From a resource-theoretic point of view, one can consider $\mathsf{MU}$ as the set of free channels and non-mixed-unitary channels as resources. Accordingly, one can then define a class of free superchannels that preserve the set $\mathsf{MU}$, i.e., map every mixed-unitary channel to another mixed-unitary channel. We consider a particularly simple class of such free superchannels and prove that all the quantum Stein exponents introduced in the previous section, when applied to this setting, are monotonic under the action of these superchannels.

More specifically, consider the following class of channel transformations, obtained by pre- and post-processing with mixed-unitary channels:
    \begin{align}
        \Theta_{\Phi_1,\Phi_2}(\mathcal{E}) =\Phi_2\circ\mathcal{E}\circ\Phi_1,
    \end{align}
where $\mathcal{E}$ is any quantum channel, and  $\Phi_1,\Phi_2\in\mathsf{MU}$. Now, since the composition of mixed-unitary channels is again mixed-unitary, it follows that if $\mathcal{E}\in\mathsf{MU}$ then $\Theta_{\Phi_1,\Phi_2}(\mathcal{E})\in\mathsf{MU}$. Therefore, these transformations preserve the set $\mathsf{MU}$ and can be regarded as free superchannels in the resource theory of non-mixed-unitary channels. This motivates the following notion of mixed-unitary-convertibility.

\begin{definition}[MU-convertible channels]
A channel $\mathcal{N}$ is said to be \emph{MU-convertible} to a channel $\mathcal{M}$ if there exist mixed-unitary channels $\Phi_1$ and $\Phi_2$ such that
\begin{align}
    \mathcal{M}=\Theta_{\Phi_1,\Phi_2}(\mathcal{N}).
\end{align}
If, in addition, $\Phi_1$ and $\Phi_2$ are unitary channels, then $\mathcal{N}$ and $\mathcal{M}$ are said to be \emph{unitarily equivalent}. 
\end{definition}

Our following theorem establishes the monotonicity of the quantum Stein exponents under the class of free superchannels considered here.

\begin{theorem}[Monotonicity of Stein exponents]\label{theo:monotonicity}
Let $\mathcal{N}$ and $\mathcal{M}$ be two channels such that $\mathcal{N}$ is MU-convertible to $\mathcal{M}$. Then, for any block size $k$, the corresponding memoryless and memory-assisted quantum Stein exponents satisfy
    \begin{align}
    E^{k}\left(\mathsf{MU}^{\mathrm{iid}}\middle\|\mathcal{M}^{\mathrm{iid}}\right)&\leq E^{k}\left(\mathsf{MU}^{\mathrm{iid}}\middle\|\mathcal{N}^{\mathrm{iid}}\right),\\
    \widetilde{E}^{k}\left(\mathsf{MU}^{\mathrm{iid}}\middle\|\mathcal{M}^{\mathrm{iid}}\right)&\leq\widetilde{E}^{k}\left(\mathsf{MU}^{\mathrm{iid}}\middle\|\mathcal{N}^{\mathrm{iid}}\right),
    \end{align}
    and the same ordering holds for the Stein exponents corresponding to block-i.i.d. probe states with any of the correlation constraints considered above in $\biid{k}$ strategy. Moreover, the inequality is saturated whenever $\mathcal{N}$ and $\mathcal{M}$ are unitarily equivalent.
\end{theorem}

\begin{proof}
    We first prove the result for the memory-assisted setting without any correlation constraint on the block-i.i.d. probe states. The memoryless case follows by taking the corresponding auxiliary reference $R$ to be trivial.

    Consider the hypothesis testing problem with the null hypothesis $\mathsf{MU}$ and alternative hypothesis $\mathcal{M}$, under block-i.i.d. probes of fixed block size $k$. Let $\sigma_{(k)}\in\mathfrak{D}(R\otimes A^{\otimes k})$ be an arbitrary joint probe state over the reference $R$ and the $k$ channel inputs.

    Since $\mathcal{N}$ is MU-convertible to $\mathcal{M}$, let $\Phi_1$ and $\Phi_2$ be two mixed unitary channels that convert $\mathcal{N}$ to $\mathcal{M}$, i.e., $\mathcal{M}=\Phi_2\circ\mathcal{N}\circ\Phi_1$. Defining 
    $\widetilde{\sigma}_{(k)}=\mathrm{id}_R\otimes\Phi^{\otimes k}_1[\sigma_{(k)}]$, we have
    \begin{align}
        \mathrm{id}_R\otimes \mathcal{M}^{\otimes k}[\sigma_{(k)}]=\mathrm{id}_R\otimes\Phi_2^{\otimes k}\left[\mathrm{id}_R\otimes\mathcal{N}^{\otimes k}[\widetilde{\sigma}_{(k)}]\right].
    \end{align}
    Therefore,
    \begin{align}
        \min_{\Phi\in\mathsf{MU}} D\left(\mathrm{id}_R\otimes \Phi^{\otimes k}[\sigma_{(k)}]\middle\|\mathrm{id}_R\otimes\mathcal{M}^{\otimes k}[\sigma_{(k)}]\right)=\min_{\Phi\in\mathsf{MU}} D\left(\mathrm{id}_R\otimes\Phi^{\otimes k}[\sigma_{(k)}]\middle\|\mathrm{id}_R\otimes\Phi_2^{\otimes k}\left[\mathrm{id}_R\otimes\mathcal{N}^{\otimes k}[\widetilde{\sigma}_{(k)}]\right]\right).\label{eq:mono_1}
    \end{align} 

    Now consider the subset $\mathsf{MU}'\subseteq\mathsf{MU}$ consisting of the mixed-unitary channels obtained by applying the particular superchannel $\Theta_{\Phi_1,\Phi_2}$ to each channel in $\mathsf{MU}$, i.e., $\mathsf{MU}':=\{\Lambda=\Phi_2\circ\Phi\circ\Phi_1:\Phi\in\mathsf{MU}\}$. Then
    \begin{align}\label{eq:mono_2}
        \min_{\Phi\in\mathsf{MU}} &D\left(\mathrm{id}_R\otimes\Phi^{\otimes k}[\sigma_{(k)}]\middle\|\mathrm{id}_R\otimes\Phi_2^{\otimes k}\left[\mathrm{id}_R\otimes\mathcal{N}^{\otimes k}[\widetilde{\sigma}_{(k)}]\right]\right)\nonumber\\
        &\leq \min_{\Lambda\in\mathsf{MU}'} D\left(\mathrm{id}_R\otimes\Lambda^{\otimes k}[\sigma_{(k)}]\middle\|\mathrm{id}_R\otimes\Phi_2^{\otimes k}\left[\mathrm{id}_R\otimes\mathcal{N}^{\otimes k}[\widetilde{\sigma}_{(k)}]\right]\right),\nonumber\\
        &=\min_{\Phi\in\mathsf{MU}} D\left(\mathrm{id}_R\otimes\Phi_2^{\otimes k}\left[\mathrm{id}_R\otimes\Phi^{\otimes k}[\widetilde{\sigma}_{(k)}]\right]\middle\|\mathrm{id}_R\otimes\Phi_2^{\otimes k}\left[\mathrm{id}_R\otimes\mathcal{N}^{\otimes k}[\widetilde{\sigma}_{(k)}]\right]\right),\nonumber\\
        &\leq \min_{\Phi\in\mathsf{MU}} D\left(\mathrm{id}_R\otimes\Phi^{\otimes k}[\widetilde{\sigma}_{(k)}]\middle\|\mathrm{id}_R\otimes\mathcal{N}^{\otimes k}[\widetilde{\sigma}_{(k)}]\right),
    \end{align}
    where the last inequality is due to the DPI. Putting Eq.~\eqref{eq:mono_2} in Eq.~\eqref{eq:mono_1}, we get
    \begin{align}
        \min_{\Phi\in\mathsf{MU}} D\left(\mathrm{id}_R\otimes \Phi^{\otimes k}[\sigma_{(k)}]\middle\|\mathrm{id}_R\otimes\mathcal{M}^{\otimes k}[\sigma_{(k)}]\right)\leq \min_{\Phi\in\mathsf{MU}} D\left(\mathrm{id}_R\otimes\Phi^{\otimes k}[\widetilde{\sigma}_{(k)}]\middle\|\mathrm{id}_R\otimes\mathcal{N}^{\otimes k}[\widetilde{\sigma}_{(k)}]\right)\forall \sigma_{(k)}.\label{eq:mono_3}
    \end{align}

    Let us denote the set of states obtained by applying the particular channel $\mathrm{id}_R\otimes\Phi_1^{\otimes k}$ to states $\sigma_{(k)}\in\mathfrak{D}(R\otimes A^{\otimes k})$ by $\mathfrak{P}\subseteq\mathfrak{D}(R\otimes A^{\otimes k})$, i.e., $\mathfrak{P}:=\left\{\widetilde{\sigma}_{(k)}=\mathrm{id}_R\otimes\Phi_1^{\otimes k}[\sigma_{(k)}]:\sigma_{(k)}\in\mathfrak{D}(R\otimes A^{\otimes k})\right\}$. Therefore, from Eq.~\eqref{eq:mono_3} we have,
    \begin{subequations}
    \begin{align}
        \sup_{\sigma_{(k)}\in\mathfrak{D}(R\otimes A^{\otimes k})}\min_{\Phi\in\mathsf{MU}} &D\left(\mathrm{id}_R\otimes \Phi^{\otimes k}[\sigma_{(k)}]\middle\|\mathrm{id}_R\otimes\mathcal{M}^{\otimes k}[\sigma_{(k)}]\right)\nonumber\\
        &\leq \sup_{\widetilde{\sigma}_{(k)}\in\mathfrak{P}}\min_{\Phi\in\mathsf{MU}} D\left(\mathrm{id}_R\otimes\Phi^{\otimes k}[\widetilde{\sigma}_{(k)}]\middle\|\mathrm{id}_R\otimes\mathcal{N}^{\otimes k}[\widetilde{\sigma}_{(k)}]\right)\label{eq:mono_4}\\
        &\leq \sup_{\sigma_{(k)}\in\mathfrak{D}(R\otimes A^{\otimes k})}\min_{\Phi\in\mathsf{MU}} D\left(\mathrm{id}_R\otimes\Phi^{\otimes k}[\sigma_{(k)}]\middle\|\mathrm{id}_R\otimes\mathcal{N}^{\otimes k}[\sigma_{(k)}]\right),
    \end{align}
    \end{subequations}
    where the last inequality is due to the fact that $\mathfrak{P}\subseteq\mathfrak{D}(R\otimes A^{\otimes k})$.
    
Hence, we prove 
    \begin{align}
        \widetilde{E}^{k}\left(\mathsf{MU}^{\mathrm{iid}}\middle\|\mathcal{M}\right)\leq \widetilde{E}^{k}\left(\mathsf{MU}^{\mathrm{iid}}\middle\|\mathcal{N}\right).
        \label{eq:mon}
    \end{align}

    Now suppose that $\mathcal{N}$ and $\mathcal{M}$ are unitarily equivalent. Then there exist unitary channels $\mathcal{U}$ and $\mathcal{V}$ such that $$\mathcal{M}=\mathcal{U}\circ\mathcal{N}\circ\mathcal{V}.$$
    Since unitary channels are invertible, with their inverse channels given by their Hilbert--Schmidt adjoints, which are also unitary, we likewise have $$\mathcal{N}=\mathcal{U}^{\dagger}\circ\mathcal{M}\circ\mathcal{V}^{\dagger}.$$
    Applying the previously established monotonicity (Eq.~\eqref{eq:mon}) in both directions yields
    \begin{align}
        \widetilde{E}^{k}\left(\mathsf{MU}^{\mathrm{iid}}\middle\|\mathcal{M}^\mathrm{iid}\right)= \widetilde{E}^{k}\left(\mathsf{MU}^{\mathrm{iid}}\middle\|\mathcal{N}^\mathrm{iid}\right).
    \end{align}

    For proving that the same monotonicity also holds for the Stein exponents corresponding to probe states with any of the correlation constraints within each block considered above, note that the pre-processing channel $\Phi_1^{\otimes k}$ acts locally on each of the $k$ channel inputs. Since such a local operation cannot generate correlations between different channel uses, the class of correlation-constrained admissible probe states is preserved under $\mathrm{id}_R\otimes\Phi_1^{\otimes k}$. Therefore, in the correlation-constrained case, the supremum on the LHS of Eq.~\eqref{eq:mono_4} is over the corresponding class of correlation-constrained states, while that on the RHS is over a subset of this class. The rest of the proof proceeds in exactly the same way.
\end{proof}

The fact that all the Stein exponents associated with a fixed channel $\mathcal{N}$ vanish whenever $\mathcal{N}\in\mathsf{MU}$, together with their monotonicity established in Theorem~\ref{theo:monotonicity}, implies that they form a family of resource monotones for non-mixed-unitarity under the class of free superchannels considered here. However, as we demonstrate subsequently, not all of these monotones are faithful: certain Stein exponents may vanish for some non-mixed-unitary channels as well.

\section{Hierarchy of Stein Exponents}
\label{sec:MU-hierarchy}
We now show that the general ordering relations among the Stein exponents corresponding to different probe-state restrictions, established in Sec.~\ref{sec:block-iid}, are indeed strict in various instances of hypothesis testing for non-mixed-unitarity. We demonstrate this through explicit examples of unital but non-mixed-unitary channels.

Before presenting our results, it is instructive to state the following lemma, which highlights an interesting relation between unital and mixed-unitary channels.

\begin{lemma}[Uhlmann's majorisation theorem~\cite{Watrous2011}]\label{lemma:Uhlmann}
    Let $\rho$ and $\sigma$ be two quantum states in $\mathfrak{D}(\mathbb{C}^d)$. The following statements are equivalent:
    \begin{enumerate}
        \item $\rho$ majorises $\sigma$ ($\rho\succ\sigma$).
        \item There exists a unital channel $\mathcal{N}$ such that $\mathcal{N}(\rho)=\sigma$.
        \item There exists a mixed-unitary channel $\Phi$ such that $\Phi(\rho)=\sigma$.
    \end{enumerate}
\end{lemma}

Lemma~\ref{lemma:Uhlmann} implies that, although the set of mixed-unitary channels is a strict subset of the set of unital channels for $d\geq 3$, the state-conversion relations associated with the two classes are identical; both are characterized by majorisation.

Our following theorem establishes that, in the absence of auxiliary memory, fully i.i.d. probe states are insufficient for distinguishing any unital channel from the set of mixed-unitary channels, even with arbitrarily many uses of the channel.

\begin{theorem}\label{theo:unital_nogo}
    Let $\mathcal{N}$ be a non-mixed-unitary channel on $\mathfrak{B}(\mathbb{C}^d)$, $d\geq 3$. Then
    \begin{align}
        E^{1}(\mathsf{MU}^{\mathrm{iid}}\|\mathcal{N}^\mathrm{iid})=0,
    \end{align}
    \emph{if and only if} $\mathcal{N}$ is unital.
\end{theorem}

\begin{proof}
 We first consider the forward direction. Let $\rho\in\mathfrak{D}(\mathbb{C}^d)$ be an arbitrary probe state that is used as the input for each channel use. Since $\mathcal{N}$ is unital, the output state $\mathcal{N}(\rho)$ satisfies $$\rho\succ\mathcal{N}(\rho).$$ Therefore, according to Lemma~\ref{lemma:Uhlmann}, there exists a mixed-unitary channel
$\Phi_{\rho}$ depending on $\rho$ satisfying $$\Phi_{\rho}(\rho)=\mathcal{N}(\rho).$$ As a result,
\begin{align}
    \min_{\Phi\in\mathsf{MU}} D(\Phi(\rho)\|\mathcal{N}(\rho))=D(\Phi_{\rho}(\rho)\|\mathcal{N}(\rho))=0.
\end{align}
Since this holds for every input state $\rho$, it follows that $E^{1}(\mathsf{MU}^{\mathrm{iid}}\|\mathcal{N}^\mathrm{iid})=0$.

Conversely, suppose $\mathcal N$ is non-unital. Then $\mathcal{N}\left(\frac{\mathbb{I}_d}{d}\right)\neq \frac{\mathbb{I}_d}{d}$. Consequently,
\begin{align}
    E^{1}(\mathsf{MU}^{\mathrm{iid}}\|\mathcal{N}^\mathrm{iid})\geq \min_{\Phi\in\mathsf{MU}} D\left(\Phi\left(\frac{\mathbb{I}_d}{d}\right)\middle\|\mathcal{N}\left(\frac{\mathbb{I}_d}{d}\right)\right)=D\left(\frac{\mathbb{I}_d}{d}\middle\|\mathcal{N}\left(\frac{\mathbb{I}_d}{d}\right)\right)>0.
\end{align}
Therefore, a vanishing exponent implies unitality of $\mathcal{N}$.
\end{proof}

This naturally motivates the question of what additional probe settings are required to obtain a nonzero Stein exponent in the non-mixed-unitarity testing of unital channels. In particular, within the memoryless setting, is increasing the block size while restricting the input probes within each block to product states sufficient, or is entanglement across different channel uses necessary? Furthermore, does access to auxiliary memory provide an advantage even when the input probes are fully i.i.d.? In the following, we address these questions for several unital but non-mixed-unitary channels. 

\subsection{Strict advantage of increasing block size with product inputs}\label{subsec:block_size_advantage}

Here we consider the setting of block-i.i.d. product probe states ($\biid{k}$) without access to auxiliary reference system, i.e., in the memoryless setting. The particular family of channels that we consider here is the family of $O(d)$-covariant channels, where $O(d)$ is the real orthogonal group in dimension $d$, defined as \[
O(d):=\left\{O\in\mathbb{R}^{d\times d}:O^{T}O=OO^{T}=\mathbb{I}_d\right\}.
\]
\paragraph*{\textbf{$O(d)$-covariant channels:}} A quantum channel $\mathcal{E}$ on $\mathfrak{B}(\mathbb{C}^d)$ is said to be $O(d)$-covariant if 
\begin{align}
    \mathcal{E}\left(OXO^{\intercal}\right)=O\mathcal{E}(X)O^{\intercal},\qquad\forall\,O\in O(d),\,X\in\mathfrak{B}(\mathbb{C}^d).
\end{align}
It follows that the family of $O(d)$-covariant channels forms a subset of the set of unital channels on $\mathfrak{B}(\mathbb{C}^d)$. Furthermore, in every dimension $d$, the set is convex and has three extreme points \cite{Keyl_2002,Mendl2009}:
\begin{itemize}
    \item the identity channel: \quad $\mathrm{id}(X)=X$,
    \item the Werner-Holevo channel: \quad $\mathcal{N}_{WH}(X)=\frac{1}{d-1}\left[\Tr(X)\,\mathbb{I}_d-X^{\intercal}\right]$,
    \item and the channel: \quad $\Psi(X)=\frac{1}{(d-1)(d+2)}\left[d\left(\Tr(X)\,\mathbb{I}_d+X^{\intercal}\right)-2X\right]$,
\end{itemize}
where $X^{\intercal}$ denotes the transpose of $X$ with respect to the computational basis. The channel $\Psi$ is mixed-unitary for every dimension $d$, whereas the Werner--Holevo channel $\mathcal{N}_{WH}$ is mixed-unitary if and only if $d$ is even. Therefore, every $O(d)$-covariant channel is mixed-unitary for even $d$, while for odd $d$ the family contains non-mixed-unitary channels \cite{Mendl2009}.

\subsubsection{Block-i.i.d. product states of block size $k=2$ may be insufficient for non-mixed-unitary testing of $O(3)$-covariant channels}
Here, we investigate whether block-i.i.d. product states of block size two can test non-mixed-unitarity of $O(3)$-covariant channels. Since the Werner–Holevo channel $\mathcal{N}_{WH}$ is the only extreme non-mixed-unitary channel in the set of $O(3)$-covariant channels and any $O(3)$-covariant channel can be expressed as a convex mixture of the extreme channels mentioned above, a direct approach is to ask whether a mixed-unitary channel can simulate the action of $\mathcal{N}_{WH}$ on an arbitrary pair of input states. In particular, let us examine whether an equal-weight, rank-two mixed-unitary channel suffices \cite{girard2022mixed}. Specifically, for unitaries $U, V \in U(3)$, where $U(n)$ denotes the set of all $n$-dimensional unitary matrices, let
\begin{equation}
\Phi_{U,V}(X) = \frac{1}{2}\left(UXU^\dagger + VXV^\dagger\right).
\label{eq:Phi_UV}
\end{equation}
We then ask whether for every pair of density operators $\{\rho_1, \rho_2\}$ acting on $\mathbb{C}^3$, there exist $U, V$, dependent on the pair, satisfying
\begin{equation}
\label{eq:WH_MU_equivalence}
\Phi_{U,V}(\rho_i) = \mathcal{N}_{WH}(\rho_i), \quad \text{for } i \in \{1, 2\}.
\end{equation}
An affirmative answer implies that block size two is insufficient to distinguish any $O(3)$-covariant non-mixed-unitary channel from the mixed-unitary set. 

Let $\mathcal{H}_0(3)$ denote the real vector space of traceless Hermitian $3\times 3$ matrices. Expressing each state as $\rho_i = \frac{\mathbb I_3}{3} + M_i$ for $i=1,2$, where $M_1,M_2\in\mathcal{H}_0(3)$, one finds that this condition is equivalent to finding $U, V \in U(3)$ such that
\begin{equation}
\label{eq:transpose_hollowization}
M_i^T + UM_iU^\dagger + VM_iV^\dagger = 0, \quad \text{for } i \in \{1, 2\}.
\end{equation}
By combining $M_1$ and $M_2$ into a single general traceless matrix $Z = M_1 + \iota M_2$ with $\iota=\sqrt{-1}$, this requirement reduces compactly to
\begin{equation}
\label{eq:transpose_hollowization_complex}
Z^T + UZU^\dagger + VZV^\dagger = 0.
\end{equation}

By Fillmore’s theorem~\cite{Fillmore1969}, every complex traceless matrix is unitarily similar to a hollow matrix (a matrix whose diagonal entries all vanish). Hence the following lemma holds.

\begin{lemma}
[Hollowization of a traceless matrix]
\label{lemma:single_Hollowization}
If $X \in \mathcal M_n(\mathbb{C})$ with $\mathcal M_n(\mathbb{C})$ being the set of all $n$-dimensional complex matrices, and $\Tr(X) = 0$, then there exists a unitary $Q \in U(n)$ such that every diagonal entry of $Q X Q^\dagger$ equals to zero.
\end{lemma}
Any general traceless matrix $Z \in \mathcal M_3(\mathbb{C})$ admits a unique Cartesian decomposition $Z = M_1 + \iota M_2$, where $M_1, M_2 \in \mathcal{H}_0(3)$ are traceless Hermitian matrices. Applying Lemma~\ref{lemma:single_Hollowization} directly to $Z$ implies that a single unitary $Q \in U(3)$ simultaneously hollows out both Hermitian components: $(Q M_j Q^\dagger)_{kk} = 0~\forall ~k$ and $j \in \{1, 2\}$. This property yields the following identity for any pair of traceless Hermitian matrices.
\begin{lemma}
\label{lemma:simul_averaging}
For every pair $M_1, M_2 \in \mathcal{H}_0(3)$, there exists a unitary $W \in U(3)$ such that
\begin{equation}
\label{eq:simul_hollowization}
M_j + W M_j W^\dagger + W^\dagger M_j W = 0, \quad j \in \{1, 2\}.
\end{equation}
\end{lemma}
\begin{proof}
By simultaneous hollowization, choose $Q \in U(3)$ such that $X_j := Q M_j Q^\dagger$ is hollow for both $j = 1, 2$. Let $\omega = e^{2\pi i / 3}$ and set $D = \operatorname{diag}(1, \omega, \omega^2)$. For the off-diagonal entries ($r \neq s$), we have
$$\left(X_j + D X_j D^\dagger + D^2 X_j (D^2)^\dagger\right)_{rs} = (X_j)_{rs}\left(1 + \omega^{r-s} + \omega^{2(r-s)}\right) = 0.$$
Because each $X_j$ is hollow, the diagonal entries also vanish identically. Hence \[X_j+DX_jD^\dagger+D^2X_j(D^2)^\dagger=0.\] 
Since $D^2=D^\dagger$, conjugating by $Q^\dagger$ and taking $W=Q^\dagger D Q$ immediately yields Eq.~\eqref{eq:simul_hollowization}.
\end{proof}
Note that Eq.~\eqref{eq:simul_hollowization} immediately implies Eq.~\eqref{eq:transpose_hollowization} whenever the pair $(M_1, M_2)$ is jointly unitarily equivalent to its transpose (joint UET), i.e., whenever there exists a single unitary $S \in U(3)$ satisfying $S M_j S^\dagger = M_j^T$ for $j \in \{1, 2\}$ \cite{garcia2012unitary}. It is therefore crucial to identify concrete conditions under which joint UET holds. For $3 \times 3$ traceless Hermitian pairs this is, in fact, governed by one exact scalar criterion. 

Towards finding the joint UET criterion for the pair $(M_1, M_2)\in \mathcal{H}_0(3)^{\times2}$, let us define the simultaneous-similar $GL(3,\mathbb{C})$-orbit of a pair of Hermitian matrices $(M_1, M_2)$ as \[\mathcal{O}(M_1,M_2)=\{(PM_1P^{-1}, PM_2P^{-1}):P\in GL(3,\mathbb{C})\}.\]
Because tuples of Hermitian matrices are completely reducible, by Artin's theorem~\cite{Artin1969,Morrison1980}, their simultaneous-similar $GL(3,\mathbb{C})$-orbit is topologically closed. Invariant theory then guarantees that, for two distinct closed orbits $\mathcal{O}(M_1,M_2)$ and $\mathcal{O}(M_1',M_2')$, there exists an invariant polynomial $I$ such that $I(M_1,M_2)\neq I(M_1',M_2')$~\cite{Mumford2002}. Conversely, two closed orbits coincide if and only if all invariant polynomials take the same values on them. Hence, two Hermitian pairs $(M_1,M_2)$ and $(M_1',M_2')$ belong to the simultaneous-similar $GL(3,\mathbb C)$-orbit if and only if all polynomial invariants under simultaneous conjugation take the same values on the two pairs. By Procesi's theorem~\cite{Procesi1976}, the invariant algebra for simultaneous conjugation of matrices is generated by traces of words in the matrices\footnote{A word of matrices is a finite product of matrices selected from a given collection, in a specified order. For the pair $(M_1,M_2)$, examples of words are $M_1, M_2, M_1M_2, M_2M_1, M_1^2M_2,M_1M_2M_1$, and so on. Since matrices generally do not commute, $M_1M_2$ and $M_2M_1$ are different words. More generally, a $k$-length word has the form $\omega(M_1,M_2)=L_1L_2...L_k$, where each $L_i\in\{M_1,M_2\}$.}. In the particular case of two traceless $3\times3$ matrices, a minimal generating set may be chosen as the following nine polynomial invariants~\cite{Aslaksen2006}:
\begin{eqnarray}
\label{eq:inv_algebra}
&&I_1(M_1,M_2)=\Tr(M_1^2),\;
I_2(M_1,M_2)=\Tr(M_1M_2),\;
I_3(M_1,M_2)=\Tr(M_2^2),\nonumber\\
&&I_4(M_1,M_2)=\Tr(M_1^3),\;
I_5(M_1,M_2)=\Tr(M_1^2M_2),\;
I_6(M_1,M_2)=\Tr(M_1M_2^2),\nonumber\\
&&I_7(M_1,M_2)=\Tr(M_2^3),\;
I_8(M_1,M_2)=\Tr\left(M_1^2M_2^2-M_1M_2M_1M_2\right),\nonumber\\
&&I_9(M_1,M_2)=\Tr\left(M_1^2M_2^2M_1M_2-M_2^2M_1^2M_2M_1\right).
\end{eqnarray}
Consequently, for two traceless Hermitian pairs $(M_1,M_2)$ and $(M_1',M_2')$, whose simultaneous-similar $GL(3,\mathbb C)$-orbits are closed,
\[\mathcal O(M_1,M_2)=\mathcal O(M_1',M_2') \quad\Longleftrightarrow\quad I_k(M_1,M_2)=I_k(M_1',M_2'),\qquad k=1,\ldots,9.\]
Thus, the above nine invariants provide a finite necessary-and-sufficient criterion for distinguishing the closed simultaneous-similar $GL(3,\mathbb C)$-orbits of traceless Hermitian matrix pairs. The following theorem provides the necessary-and-sufficient condition for a pair of traceless $3\times3$ Hermitian matrices to be joint UET.
\begin{theorem}
    \label{th:joint-UET}
    For traceless $3\times3$ Hermitian matrices $M_1, M_2$, the following are equivalent:
    \begin{itemize}
        \item $(M_1,M_2)$ is joint UET;
        \item $\chi(M_1,M_2)=\Tr\left(M_1^2M_2^2M_1M_2-M_2^2M_1^2M_2M_1\right)=0.$
    \end{itemize}
    Consequently, every pair with $\chi(M_1,M_2)=0$ satisfies Eq.~\eqref{eq:transpose_hollowization}. 
\end{theorem}
\begin{proof}
The simultaneous-similarity invariant algebra of two traceless $3 \times 3$ matrices is generated by the nine quantities given in Eq.~\eqref{eq:inv_algebra}. Transposition preserves the first eight generators. Indeed, transposing a word reverses its letters; for the displayed words, cyclic property of trace returns the original expression. On the other hand, reversal interchanges the two degree-six words in $\chi(M_1,M_2)$, and therefore $\chi(M_1^T, M_2^T) = -\chi(M_1, M_2)$. If $(M_1, M_2)$ is joint UET, simultaneous unitary similarity preserves every trace invariant. Therefore, $\chi(M_1, M_2) = -\chi(M_1, M_2)$, which implies $\chi(M_1, M_2) = 0$.

Conversely, suppose $\chi(M_1, M_2) = 0$. Then all nine generating invariants of $(M_1, M_2)$ and $(M_1^T, M_2^T)$ agree. Invariant polynomials separate closed simultaneous $GL(3,\mathbb{C})$-orbits. A Hermitian tuple generates a finite-dimensional $*$-algebra and is completely reducible, so its complex similarity orbit is closed. It follows that an invertible matrix $P$ satisfies
\begin{equation}
\label{eq:UET}
    PM_1P^{-1} = M_1^T, \qquad PM_2P^{-1} = M_2^T.
\end{equation}

It remains to upgrade $P$ to a unitary. Taking adjoints in Eq.~\eqref{eq:UET}, using Hermiticity, and comparing with Eq.~\eqref{eq:UET}, shows
\begin{equation*}
    (P^\dagger P)M_1 = M_1(P^\dagger P), \qquad (P^\dagger P)M_2 = M_2(P^\dagger P).
\end{equation*}
Let $T = (P^\dagger P)^{1/2}$. Then $T$ commutes with both $M_1$ and $M_2$. In the polar decomposition $P = ST$, the matrix $S = PT^{-1}$ is unitary and
\begin{equation}
\label{eq:UET_final}
    SM_1S^\dagger = PM_1P^{-1} = M_1^T, \qquad SM_2S^\dagger = PM_2P^{-1} = M_2^T.
\end{equation}
Thus $(M_1, M_2)$ is joint UET. Substituting $M_j = S^\dagger M_j^T S$ into Eq.~\eqref{eq:simul_hollowization} and identifying $U = SW$ and $V = SW^\dagger$ immediately yields Eq.~\eqref{eq:transpose_hollowization}.
\end{proof}
We now examine several notable configurations where the joint UET condition is satisfied.
\begin{enumerate}
    \item \textbf{Both states are real.} Suppose the density matrices are expressed in the basis in which the transposition map is defined. Because $\rho_i=\frac{\mathbb I_3}{3}+M_i$, real $\rho_i$ implies real $M_i$ for $i=1,2$. Hermiticity then ensures that $M_i=M_i^T$ for $i=1,2$. Hence, Eq.~\eqref{eq:UET_final} holds trivially with $S=\mathbb I_3$. 
    \item \textbf{The two states commute.} Notice that $[\rho_1,\rho_2]=0\iff [M_1,M_2]=0$. For any commuting pair $(M_1, M_2)\in \mathcal{H}_0(3)^{\times2}$, there exists a unitary $Q\in U(3)$ and real diagonal matrices $D_1,D_2$, such that $M_j=QD_jQ^\dagger,~j=1,2$. Setting $S = \bar{Q} Q^\dagger$, it follows that $SM_jS^\dagger=\bar{Q}D_jQ^T=M_j^T$ for both $j$. 
    \item \textbf{One density matrix has a repeated eigenvalue.} If either $\rho_1$ or $\rho_2$ has an eigenvalue of multiplicity at least two, then the pair is joint UET.
    \begin{proof}
     Let us first note that each eigenvalue of $\rho_i$ differs from the corresponding eigenvalue of $M_i$ by $1/3$, therefore assuming $\rho_1$ has repeated eigenvalue implies the same for $M_1$. Hence, after a suitable unitary conjugation $M_j'=QM_jQ^\dagger$, we can write 
    \begin{equation*}
       M_1'=\begin{pmatrix}
           \lambda & 0\\0 & \mu I_2
       \end{pmatrix};\quad
        M_2'=\begin{pmatrix}
           r & b^\dagger\\b & C
       \end{pmatrix}, 
    \end{equation*}
     where $r\in\mathbb{R}$, $b\in\mathbb{C}^2$, and $C=C^\dagger\in \mathcal M_2$. A direct computation then shows that $\chi(M_1',M_2')=\Tr(M_1'^2M_2'^2M_1'M_2'-M_2'^2M_1'^2M_2'M_1')=0$. Therefore, by Theorem~\ref{th:joint-UET}, the pair $(M_1',M_2')$ is joint UET, and consequently the same holds for $(M_1,M_2)$.
    \end{proof}
    \textit{An immediate consequence is that joint UET holds whenever either $\rho_1$ or $\rho_2$ is a pure state.}
    \item \textbf{The pair $(\rho_1,\rho_2)$ has a common eigenvector.} If $\rho_1$ and $\rho_2$ have a common eigenvector, then the pair is joint UET.
    \begin{proof}
        Since $M_i$ and $\rho_i$ have exactly the same set of eigenvectors, $M_1$ and $M_2$ also possess a common eigenvector. Because both matrices are Hermitian, the orthogonal complement of a common eigenvector is invariant for both. Hence the pair has the simultaneous block form \[ M_j=\alpha_j\oplus B_j, \qquad B_j=B_j^\dagger\in M_2.\]
        Write the \(2\times 2\) blocks in Bloch form \[B_j=\beta_j \mathbb I_2+x_j\sigma_x+y_j\sigma_y+z_j\sigma_z.\]

        Transposition reflects the Bloch vector \( r_j=(x_j,y_j,z_j) \) to \(r_j'=(x_j,-y_j,z_j).\) The two ordered pairs of vectors \((r_1,r_2)\) and \((r_1',r_2')\) have the same Gram matrix. Therefore there is an orthogonal transformation \(R\in SO(3)\) sending \(r_j\) to \(r_j'\) for \(j=1,2\): define the isometry on their span and choose its sign on the orthogonal complement so that the determinant is \(+1\). The adjoint representation \( SU(2)\longrightarrow SO(3)\) is onto, so some \(T\in SU(2)\) implements \(R\) and obeys \[ T B_j T^\dagger = B_j^T, \qquad j=1,2. \]
        The unitary \( S=1\oplus T \) therefore implements simultaneous transposition of the original pair.
    \end{proof}
\end{enumerate}
Joint UET thus rules out several natural candidate probe pairs for testing the non-mixed-unitarity of $O(3)$-covariant channels:
\begin{proposition}
\label{prop:useless_probes}
A block-i.i.d. product probe $\rho_1 \otimes \rho_2$ of block size two gives a vanishing Stein exponent for any $O(3)$-covariant non-mixed-unitary channel from the set of mixed-unitary channels if the pair $(\rho_1, \rho_2)$ satisfies any of the following conditions:
\begin{enumerate}
\item Both states are real in the reference basis defining the transposition map.
\item The two states commute, or more generally, share a common eigenvector.
\item At least one state has a degenerate eigenvalue (which includes the case where either probe state is pure).
\end{enumerate}
\end{proposition}
An alternate proof of the proposition for pure product states is given in Appendix~\ref{appa}. Importantly, joint UET is only a sufficient condition for satisfying Eq.~\eqref{eq:transpose_hollowization}, not a necessary one. As an explicit counterexample, let $F$ denote the $3 \times 3$ discrete Fourier transform matrix with entries $F_{jk} = \frac{1}{\sqrt{3}}\omega^{jk}$ for $j, k \in \{0, 1, 2\}$, where $\omega = e^{2\pi \iota/3}$, and let $P = \text{diag}(1, 1, \omega)$. Define the unitaries $U = PF$ and $V = P^2 F P F$, and consider the traceless Hermitian matrices
\begin{equation}
K_1 = \begin{pmatrix}
0 & 0 & \frac{1}{2} + \frac{\sqrt{3}}{2}\iota \\
0 & 1 & 1 \\
\frac{1}{2} - \frac{\sqrt{3}}{2}\iota & 1 & -1
\end{pmatrix}, \qquad
K_2 = \begin{pmatrix}
0 & 0 & \frac{\sqrt{3}}{2} - \frac{1}{2}\iota \\
0 & -\sqrt{3} & \iota \\
\frac{\sqrt{3}}{2} + \frac{1}{2}\iota & -\iota & \sqrt{3}
\end{pmatrix}.
\end{equation}
A straightforward calculation confirms that the pair satisfies
\begin{equation}
K_j^T + U K_j U^\dagger + V K_j V^\dagger = 0, \quad j \in \{1, 2\}.
\end{equation}
However, evaluating the scalar invariant yields $\chi(K_1, K_2) = 24\iota \neq 0$. By Theorem~\ref{th:joint-UET}, the pair $(K_1, K_2)$ is not joint UET, establishing that joint UET is strictly stronger than the condition required for Eq.~\eqref{eq:transpose_hollowization} to hold.

In fact, numerical simulations strongly suggest an even broader conclusion. We have generated $5000$ random pairs $(\rho_1, \rho_2)$ sampled uniformly with respect to the volume measure induced by the Hilbert–Schmidt metric, then constructed $M_i = \rho_i - \mathbb{I}_3/3$, and set $Z = M_1 + \iota M_2$. For each sample, we have numerically minimized the Frobenius norm $\|Z^T + UZU^\dagger + VZV^\dagger\|_F$ over all pairs of unitaries $U,V \in U(3)$. In every instance, the objective function has reached zero within a numerical tolerance of the order of $10^{-10}$. A Mathematica notebook detailing these computations is available in Ref.~\cite{AP2026repo}. These findings provide compelling evidence that Eq.~\eqref{eq:transpose_hollowization} holds universally for all traceless Hermitian pairs $M_1, M_2 \in \mathcal{H}_0(3)$. We therefore formulate the following conjecture:
\begin{conjecture}
\label{conj:rank2_simulation}
For every pair of density operators $\{\rho_1, \rho_2\}$ acting on $\mathbb{C}^3$, there exist unitaries $U, V \in U(3)$ such that
\begin{equation*}
\Phi_{U,V}(\rho_i) = \mathcal{N}_{WH}^\mathsf{iid}(\rho_i), \quad i \in \{1, 2\},
\end{equation*}
where $\Phi_{U,V}(X) = \frac{1}{2}(UXU^\dagger + VXV^\dagger)$ is an equal-weight, rank-two mixed-unitary channel. Consequently, block-i.i.d. product states of block size $k=2$ are insufficient for testing non-mixed-unitarity of $O(3)$-covariant channels, i.e., $E_{\mathrm{Pro}}^2(\mathsf{MU}^\mathrm{iid}\|\mathcal N)=0$ where $\mathcal N=q_1 \mathcal{N}_{WH}+q_2 \Psi+q_3 \mathrm{id},$ where $\{q_i\}$ is a probability vector.
\end{conjecture}

Combining Conjecture~\ref{conj:rank2_simulation} and Theorem~\ref{theo:monotonicity}, we have the following corollary:
\begin{corollary}[Assuming Conjecture~\ref{conj:rank2_simulation}]
\label{cor:e2prozero}
    Let $\bar{\mathcal{N}}$ be a non-mixed-unitary channel that is MU-convertible from an $O(3)$-covariant channel.
    Then $E^{2}_{\mathrm{Pro}}\left(\mathsf{MU}^{\mathrm{iid}}\|\bar{\mathcal{N}}^{\mathrm{iid}}\right)=0$.
\end{corollary}
\subsubsection{Sufficiency of block-i.i.d. product states of block size $3$ for non-mixed-unitary testing of the qutrit Werner--Holevo channel}

We now investigate whether increasing the block size from $k=2$ to $k=3$ can yield a strictly positive Stein exponent. The following theorem answers the question of sufficiency of $\biid{3}$ strategy affirmatively for the qutrit Werner--Holevo channel and, consequently, for all channels unitarily equivalent to it.

\begin{theorem}
\label{th:3copiesproduct}
    Let $\mathcal{N}$ be a quantum channel on $\mathfrak{B}(\mathbb{C}^3)$ that is unitarily equivalent to the Werner--Holevo channel $\mathcal{N}_{WH}$. Then
    \begin{align}
        E_{\mathrm{Pro}}^{3}(\mathsf{MU}^{\mathrm{iid}}\|\mathcal{N}^{\mathrm{iid}})>0.
    \end{align}
\end{theorem}

\begin{proof}
We prove that $E^3_{\mathrm{Pro}}(\mathsf{MU}^{\mathrm{iid}}\|\mathcal{N}_{WH})>0$. The same conclusion then follows by Theorem~\ref{theo:monotonicity} for every channel that is unitarily equivalent to $\mathcal{N}_{WH}$. By Lemma~\ref{usefulcond}, it suffices to show that there exists a set of three states $\{\rho_1,\rho_2,\rho_3\}\subset\mathfrak{D}(\mathbb{C}^3)$ for which no mixed-unitary channel $\Phi$ satisfies
\begin{align}\label{eq:three_state_simulation}
    \Phi(\rho_i)=\mathcal{N}_{WH}(\rho_i)
    =\frac{1}{2}\left(\mathbb{I}_3-\rho_i^{\intercal}\right),
    \qquad i=1,2,3,
\end{align}
simultaneously.

The proof is constructive. Consider the following three pure states in $\mathbb{C}^3$:
\begin{align}
    \ket{\psi_1}&=\ket{0},\nonumber\\
    \ket{\psi_2}&=\frac{1}{2}\ket{0}
    +\frac{\sqrt{3}}{2}\ket{1},\nonumber\\
    \ket{\psi_3}&=\frac{1}{2}\ket{0}
    +e^{\iota\pi/3}\frac{\sqrt{3}}{2}\ket{1},
\end{align}
where $\iota=\sqrt{-1}$, and let $\rho_i=\ketbra{\psi_i}{\psi_i}$.

We prove this by contradiction. Assume that there exists a mixed-unitary channel
\begin{align}
    \Phi(X)=\sum_x p_x U_x XU_x^\dagger
\end{align}
satisfying Eq.~\eqref{eq:three_state_simulation}, where we retain only terms with $p_x>0$. Since
\begin{align}
    \mathcal{N}_{WH}(\rho_i)
    =\frac{1}{2}\left(\mathbb{I}_3
    -\ketbra{\psi_i^*}{\psi_i^*}\right),
\end{align}
we have
\begin{align}
    0
    &=\bra{\psi_i^*}\mathcal{N}_{WH}(\rho_i)\ket{\psi_i^*}
    \nonumber\\
    &=\sum_x p_x
    \left|\bra{\psi_i^*}U_x\ket{\psi_i}\right|^2,
    \qquad i=1,2,3.
\end{align}
Since, every term in the sum is nonnegative, we can write
\begin{align}\label{eq:three_support_constraints}
    \bra{\psi_i^*}U_x\ket{\psi_i}
    =\bra{\psi_i^{\intercal}}U_x\ket{\psi_i}=0,
    \qquad i=1,2,3,\quad\forall x.
\end{align}

Let $S=\mathrm{span}\{\ket{0},\ket{1}\}$. Relative to the decomposition
$\mathbb{C}^3=S\oplus\mathrm{span}\{\ket{2}\}$, we write
\begin{align}
    U_x=
    \begin{pmatrix}
        A_x & B_x\\
        C_x^\dagger & u_x
    \end{pmatrix},
\end{align}
where $A_x$ is a $2\times2$ matrix acting on the subspace \(S\), $B_x$ and $C_x$ are $2\times1$ matrices, and $u_x\in\mathbb{C}$.

The three input vectors are proportional to $(1,z_i,0)^{\intercal}$, where
\begin{align}
    z_1=0,\qquad z_2=\sqrt{3},\qquad
    z_3=\sqrt{3}e^{\iota\pi/3}.
\end{align}
Consequently, Eq.~\eqref{eq:three_support_constraints} implies
\begin{align}
    (A_x)_{00}
    +\left[(A_x)_{01}+(A_x)_{10}\right]z_i
    +(A_x)_{11}z_i^2=0,
    \qquad i=1,2,3.
\end{align}
Since a polynomial of degree at most two cannot have three distinct roots unless it vanishes identically, we get
\begin{align}
    (A_x)_{00}=(A_x)_{11}=0,
    \qquad (A_x)_{01}=-(A_x)_{10}.
\end{align}
Therefore,
\begin{align}
    A_x=
    \begin{pmatrix}
        0&a_x\\
        -a_x&0
    \end{pmatrix},
    \qquad a_x\in\mathbb{C}.
\end{align}

The unitarity of $U_x$ implies
\begin{align}
    A_xA_x^\dagger+B_xB_x^\dagger
    &=\mathbb{I}_2,\nonumber\\
    \implies B_xB_x^\dagger
    &=(1-|a_x|^2)\mathbb{I}_2.
\end{align}
Since $\mathrm{Rank}(B_xB_x^\dagger)\leq1$, the last equality requires
\begin{align}
    |a_x|=1,\qquad B_x=0.
\end{align}
Furthermore, using $U_x^\dagger U_x=\mathbb{I}_3$, we have
\begin{align}
    A_x^\dagger A_x+C_xC_x^\dagger
    &=\mathbb{I}_2,\nonumber\\
    \implies C_xC_x^\dagger&=0,\nonumber\\
    \implies C_x&=0.
\end{align}
Thus, every unitary $U_x$ preserves the subspace $S$. Since all three input states belong to $S$, it follows that
\begin{align}
    \bra{2}\Phi(\rho_i)\ket{2}
    =\sum_x p_x\left|\bra{2}U_x\ket{\psi_i}\right|^2
    =0,\qquad i=1,2,3.
\end{align}
On the other hand,
\begin{align}
    \bra{2}\mathcal{N}_{WH}(\rho_i)\ket{2}
    =\frac{1}{2}\bra{2}
    \left(\mathbb{I}_3-\ketbra{\psi_i^*}{\psi_i^*}\right)
    \ket{2}
    =\frac{1}{2},
    \qquad i=1,2,3,
\end{align}
which contradicts Eq.~\eqref{eq:three_state_simulation}.

Therefore, no mixed-unitary channel $\Phi$ reproduces the action of $\mathcal{N}_{WH}$ simultaneously on these three input states. By Lemma~\ref{usefulcond}, we then conclude that
\begin{align}
    E^3_{\mathrm{Pro}}
    \left(\mathsf{MU}^{\mathrm{iid}}\middle\|\mathcal{N}_{WH}^{\mathrm{iid}}\right)>0.
\end{align}
This completes the proof.
\end{proof}

The findings of this subsection can be summarised as follows. While Uhlmann's majorisation theorem (Lemma~\ref{lemma:Uhlmann}) guarantees that the action of any unital channel on a single state can always be reproduced by a mixed-unitary channel, our numerical evidence, together with analytical results for several special cases, suggests that the action of the qutrit Werner--Holevo channel on any two uncorrelated states can likewise be reproduced simultaneously by a single mixed-unitary channel, with the latter depending on the pair of input states. In contrast, such a simultaneous simulation is impossible, in general, for three input states. This suggests that three uncorrelated input states may be necessary to distinguish the Werner--Holevo channel from the set of mixed-unitary channels.

The resulting hierarchy of Stein exponents with respect to the block size for product probe states is given by:
\begin{align}
    &E_{\mathrm{Pro}}^{2}(\mathsf{MU}^{\mathrm{iid}}\|\mathcal{N}_{WH}^{\mathrm{iid}})\overset{?}{=}E^{1}(\mathsf{MU}^{\mathrm{iid}}\|\mathcal{N}_{WH}^{\mathrm{iid}})=0,\nonumber\\
    \text{whereas}\quad&E_{\mathrm{Pro}}^{3}(\mathsf{MU}^{\mathrm{iid}}\|\mathcal{N}_{WH}^{\mathrm{iid}})>0.
\end{align}
Thus, increasing the block size provides a strict advantage even when the probe states remain fully product across the channel uses within each block.

\subsection{Advantage of entangled probes}\label{subsec:entanglement_advantage}
For block size $k=2$, while block-i.i.d. product states may not be sufficient to distinguish an $O(3)$-covariant channel from the set of mixed-unitary channels, we show below that allowing entanglement between inputs of the two channel uses is, in general, advantageous. In particular, the Stein exponent without auxiliary memory corresponding to block-i.i.d. probes without correlation constraints is strictly positive for any channel that is unitarily equivalent to the Werner--Holevo channel.

\begin{theorem}\label{theo:E2ent}
    Let $\mathcal{N}$ be a quantum channel on $\mathfrak{B}(\mathbb{C}^3)$ that is unitarily equivalent to the Werner--Holevo channel $\mathcal{N}_{WH}$. Then
    \begin{align}
        E^2(\mathsf{MU}^{\mathrm{iid}}\|\mathcal{N}^\mathsf{iid})>0.
    \end{align}
\end{theorem}

\begin{proof}
    Due to unitary equivalence it suffices to prove that $E^{2}(\mathsf{MU}^{\mathrm{iid}}\|\mathcal{N}_{WH})>0$. By Lemma~\ref{usefulcond}, we need to show that there exists an entangled state $\rho_{AB}\in \mathfrak{D}(\mathbb{C}^3\otimes\mathbb{C}^{3})$ for which no mixed-unitary channel $\Phi$ exists satisfying
    \begin{align}\label{eq:twosiment}
        \Phi\otimes\Phi(\rho_{AB})=\mathcal{N}_{WH}\otimes\mathcal{N}_{WH} (\rho_{AB}).
    \end{align}

    We prove this by contradiction. Consider the entangled state $\rho^{\psi}_{AB}=\ketbra{\psi}{\psi}\in\mathfrak{D}(\mathbb{C}^3\otimes\mathbb{C}^{3})$ as input to the two channel uses, where
    \begin{align}
        \ket{\psi}_{AB}=\frac{1}{\sqrt{2}}(\ket{01}_{AB}-\ket{10}_{AB})\in\mathbb{C}^3\otimes\mathbb{C}^{3}.
    \end{align}
    Since $\left(\rho^{\psi}_{AB}\right)^{\intercal}=\rho^{\psi}_{AB}$, the action of the Werner--Holevo channels on this state is given by
    \begin{align}
        \sigma_{AB}=\mathcal{N}_{WH}\otimes\mathcal{N}_{WH}(\ketbra{\psi}{\psi})=\frac{1}{4}\left(\mathbb{I}_A\otimes\mathbb{I}_B-\mathbb{I}_A\otimes \frac{P_B}{2}-\frac{P_A}{2}\otimes \mathbb{I}_B+\rho^{\psi}_{AB}\right),
    \end{align}
    where $P=\ketbra{0}{0}+\ketbra{1}{1}$ denotes the projector onto the local subspace $S=\mathrm{span}\{\ket{0},\ket{1}\}$. 
    
    Consider the projection of $\sigma_{AB}$ onto the subspace $S_A\otimes S_B$, which is given by
    \begin{align}
        P_A\otimes P_B\,(\sigma_{AB})\,P_A\otimes P_B&=\frac{1}{4}\left(P_A\otimes P_B-P_A\otimes \frac{P_B}{2}-\frac{P_A}{2}\otimes P_B+\rho^{\psi}_{AB}\right),\nonumber\\
        &=\frac{1}{4}\rho^{\psi}_{AB}=\frac{1}{4}\ketbra{\psi}{\psi}.\label{eq:entcomp}
    \end{align}
    Therefore, $\sigma_{AB}$ has a nonzero rank-one entangled component inside the subspace $S_A\otimes S_B$.

    Now assume that there exists a mixed-unitary channel $\Phi$ satisfying Eq.~\eqref{eq:twosiment} for the considered input state. We therefore have 
    \begin{align}
        \Phi(\rho_{A(B)})=\mathcal{N}_{WH}(\rho_{A(B)})=\frac{1}{2}\left(\mathbb{I}_{A(B)}-\frac{P_{A(B)}}{2}\right),
    \end{align}
    where $\rho_{A(B)}=\Tr_{B(A)}[\rho^{\psi}_{AB}]=\frac{1}{2}P_{A(B)}$ denotes the marginal states of the input. Then it follows that 
    \begin{align}
        \Phi(\ketbra{2}{2}_{A(B)})&=\Phi(\mathbb{I}_{A(B)}-P_{A(B)}),\nonumber\\
        &=\mathbb{I}_{A(B)}-\Phi(P_{A(B)}),\nonumber\\
        &=\frac{1}{2}P_{A(B)}.
    \end{align}
    Considering $\Phi(X)=\sum_x p_xU_xXU_x^{\dagger}$, we thus get
    \begin{align}
        \bra{2}\Phi(\ketbra{2}{2})\ket{2}&=\sum_{x}p_x\abs{\bra{2}U_x\ket{2}}^2=\frac{1}{2}\bra{2}P_{A(B)}\ket{2}=0\nonumber\\
    \implies(U_x)_{22}&=\bra{2}U_x\ket{2}=0\quad\forall\,x.
    \end{align}
    That is, every unitary $U_x$ in the mixed-unitary channel $\Phi$ has a block form
    \begin{equation}\label{eq:u_k}
        U_x=\begin{pmatrix}
            A_x&B_x\\
            C^{\dagger}_x&0
            \end{pmatrix},
    \end{equation}
    where $A_x$ is a $2\times2$ matrix acting on the subspace $S$, and $B_x$ and $C_x$ are $2\times1$ matrices. The unitarity of $U_x$ implies
    \begin{subequations}
    \begin{align}
        B_x^{\dagger}B_x&=1,\\
        A_x^\dagger B_x&=0,\\ A_xA_x^\dagger+B_xB_x^\dagger&=\mathbb{I}_2.
    \end{align}    
    \end{subequations}
    The first equality implies $B_x\neq0$ and the second equality implies that $B_x$ lies in the kernel of $A_x^{\dagger}$. Therefore, $\mathrm{Rank}(A_x)\leq 1$. Furthermore, since $\mathrm{Rank}(B_x)= 1$, the third equality implies $A_x\neq0$. Therefore,
    \begin{align}
        \mathrm{Rank}(A_x)=1,\quad\forall\,x.
    \end{align}

    Considering the projection of $\Phi\otimes \Phi(\rho^{\psi}_{AB})$ onto the subspace $S_A\otimes S_B$, we have
    \begin{align}
        P_A\otimes P_B\left(\sum_{x,y}p_xp_y\,(U_x\otimes U_y)\rho^{\psi}_{AB}(U^{\dagger}_x\otimes U^{\dagger}_y)\right)P_A\otimes P_B=\sum_{x,y}p_xp_y\,\ketbra{\phi_{x,y}}{\phi_{x,y}},
    \end{align}
    where 
    \begin{align}
        \ket{\phi_{x,y}}&=(P_A\otimes P_B)(U_x\otimes U_y)\ket{\psi},\nonumber\\
        &=(A_x\otimes A_y)\ket{\psi}.
    \end{align}
    Since $A_x$ and $A_y$ are rank-one matrices, we can write
    \begin{align}
        A_x=\ketbra{a_x}{b_x},\quad\text{and}\quad A_y=\ketbra{e_y}{f_y},
    \end{align}
    where $\ket{a_x}$ and $\ket{b_x}$ are unnormalized vectors in $S_A$, and $\ket{e_y}$ and $\ket{f_y}$ are unnormalized vectors in $S_B$, respectively. Thus, we have 
    \begin{align}
        \ket{\phi_{x,y}}=\braket{b_xf_y}{\psi} \ket{a_x}\otimes\ket{e_y}.
    \end{align}
    This implies that every nonzero vector $\ket{\phi_{x,y}}$ is a product vector, and 
    \begin{align}
        \omega_{AB}:=\sum_{x,y}p_xp_y\,\ketbra{\phi_{x,y}}{\phi_{x,y}}
    \end{align}
    is a separable positive operator.

    On the other hand, the assumed equality in Eq.~\eqref{eq:twosiment}, together with Eq.~\eqref{eq:entcomp} requires
    \begin{align}
        \omega_{AB}=\frac{1}{4}\rho^{\psi}_{AB}=\frac{1}{4}\ketbra{\psi}{\psi}.
    \end{align}
    Since $\ketbra{\psi}{\psi}$ is nonzero and entangled, this condition cannot be satisfied. Therefore, no mixed-unitary channel $\Phi$ satisfying Eq.~\eqref{eq:twosiment} exists for the considered entangled input.
\end{proof}

Proposition~\ref{prop:useless_probes} and Theorem~\ref{theo:E2ent} establish an unconditional separation between pure product probes and pure entangled probes at block size two. The separation from arbitrary product probes, including mixed-state products, remains conditional on Conjecture~\ref{conj:rank2_simulation}.

\subsubsection{More entanglement is not necessarily more advantageous}

While Theorem~\ref{theo:E2ent} tells that entanglement between two channel uses is sufficient to distinguish quantum channels that are unitarily equivalent to the Werner--Holevo channel, it is important to note that not all bipartite entangled states offer this advantage. More interestingly, as we show in the following, this advantage does not, in general, scale with the amount of entanglement.

Consider the case where the bipartite state used as input to the two channel uses is maximally entangled. Then the corresponding Stein exponent involves optimization over only the set of maximally entangled states in Eq.~\eqref{gen_stein}, i.e., 
\begin{align}
    E^2_{\mathrm{MES}}\left(\mathsf{MU}^{\mathrm{iid}}\middle\|\mathcal{N}^\mathsf{iid}\right)=\sup_{\ket{\Omega_V}\in\mathfrak{M}}\,\min_{\Phi\in\mathsf{MU}}\frac{1}{2}D\left(\Phi^{\otimes2}(\ketbra{\Omega_V}{\Omega_V})\middle\|\mathcal{N}^{\otimes2}(\ketbra{\Omega_V}{\Omega_V})\right),
\end{align}
where 
\begin{align}
    \mathfrak{M}=\left\{\ket{\Omega_V}=(\mathbb{I}\otimes V)\ket{\Omega}:V\in U(3)\right\},
\end{align}
denotes the set of maximally entangled two-qutrit states, with $\ket{\Omega}=\frac{1}{\sqrt{3}}\sum_{i=0}^2\ket{ii}$ and $U(3)$ denoting the set of unitary operators on $\mathbb{C}^3$.

\begin{proposition}
    Let $\mathcal{N}$ be a quantum channel on $\mathfrak{B}(\mathbb{C}^3)$ that is unitarily equivalent to the qutrit Werner--Holevo channel $\mathcal{N}_{WH}$. Then
    \begin{align}
        E^2_{\mathrm{MES}}\left(\mathsf{MU}^{\mathrm{iid}}\middle\|\mathcal{N}^\mathsf{iid}\right)=0.
    \end{align}
\end{proposition}

\begin{proof}
    Due to unitary equivalence, it suffices to prove this for the qutrit Werner--Holevo channel $\mathcal{N}_{WH}$. Let $\ket{\Omega_V}=(\mathbb{I}\otimes V)\ket{\Omega}$ be an arbitrary maximally entangled two-qutrit state. We show that there exists a mixed-unitary channel $\Phi_V$ for each $\ket{\Omega_V}$ such that
    \begin{align}
        \Phi_V^{\otimes 2}(\Omega_V)=\mathcal{N}^{\otimes2}_{WH}(\Omega_V),
    \end{align}
    where $\Omega_V\equiv\ketbra{\Omega_V}{\Omega_V}$. 
    
    The action of the Werner--Holevo channels on this state is given by
    \begin{align}
        \mathcal{N}_{WH}\otimes\mathcal{N}_{WH}(\Omega_V)&=\frac{1}{4}\left(\mathbb{I}\otimes\mathbb{I}-\mathbb{I}\otimes\frac{\mathbb{I}}{3}-\frac{\mathbb{I}}{3}\otimes\mathbb{I}+\Omega_V^{\intercal}\right),\nonumber\\
        &=\frac{1}{12}\mathbb{I}\otimes\mathbb{I}+\frac{1}{4}\Omega^{\intercal}_V.
    \end{align}

    Now consider the unitary channel $\mathcal{V}^*(X)=V^*XV^{\intercal}$. Since $\ket{\Omega}$ satisfies
    \begin{align}
        (A\otimes B)\ket{\Omega}=(\mathbb{I}\otimes BA^{\intercal})\ket{\Omega}
    \end{align}
    for all operators $A$ and $B$, the action of $V^*\otimes V^*$ on $\ket{\Omega_V}$ is given by
    \begin{align}
        (V^*\otimes V^*)\ket{\Omega_V}&=(V^*\otimes V^*V)\ket{\Omega},\nonumber\\
        &=(\mathbb{I}\otimes V^*VV^{\dagger})\ket{\Omega},\nonumber\\
        &=\ket{\Omega^*_V}.
    \end{align}
    Therefore, the action of the channel $\mathcal{V}^*\otimes\mathcal{V}^*$ on $\Omega_V$ is given by
    \begin{align}
        \mathcal{V}^*\otimes\mathcal{V}^*(\Omega_V)=\Omega_V^{\intercal}.
    \end{align}
    Note that since both the marginal states of $\Omega_V$ are $\frac{\mathbb{I}}{3}$ and $\mathcal{V}^*$ is a unitary channel, both the marginal states of $\Omega_V^{\intercal}$ are also $\frac{\mathbb{I}}{3}$.

    Now consider the qutrit depolarizing channel
    \begin{align}
        \mathcal{D}_{\frac{1}{2}}(X)=\frac{1}{2}\left(\Tr(X)\,\frac{\mathbb{I}}{3}+X\right),
    \end{align}
    which is a mixed-unitary channel. For any bipartite state $\rho$ with both marginals $\frac{\mathbb{I}}{3}$, the action of $\mathcal{D}_{\frac{1}{2}}\otimes \mathcal{D}_{\frac{1}{2}}$ on it is given by
    \begin{align}
        \mathcal{D}_{\frac{1}{2}}\otimes\mathcal{D}_{\frac{1}{2}}(\rho)=\frac{1}{12}\mathbb{I}\otimes\mathbb{I}+\frac{1}{4}\rho.
    \end{align}
    Define
    \begin{align}
        \Phi_V=\mathcal{D}_{\frac{1}{2}}\circ\mathcal{V}^*,
    \end{align}
    which is also a mixed-unitary channel. The action of $\Phi_V\otimes\Phi_V$ on $\Omega_V$ is then given by
    \begin{align}
        \Phi_V\otimes\Phi_V(\Omega_V)&=\mathcal{D}_{\frac{1}{2}}\otimes\mathcal{D}_{\frac{1}{2}}\left[\mathcal{V}^*\otimes\mathcal{V}^*(\Omega_V)\right],\nonumber\\
        &=\mathcal{D}_{\frac{1}{2}}\otimes\mathcal{D}_{\frac{1}{2}}(\Omega_V^{\intercal}), \nonumber\\
        &=\frac{1}{12}\mathbb{I}\otimes\mathbb{I}+\frac{1}{4}\Omega^{\intercal}_V=\mathcal{N}_{WH}\otimes\mathcal{N}_{WH}(\Omega_V).
    \end{align}
    
    Therefore, for every maximally entangled state $\ket{\Omega_V}$,
    \begin{align}
        \min_{\Phi\in\mathsf{MU}} D(\Phi^{\otimes2}(\Omega_V)\|\mathcal{N}^{\otimes2}_{WH}(\Omega_V))=0,
    \end{align}
    and hence 
    \begin{align}
        E^2_{\mathrm{MES}}\left(\mathsf{MU}^{\mathrm{iid}}\middle\|\mathcal{N}_{WH}^\mathsf{iid}\right)=0.
    \end{align}
This completes the proof.
\end{proof}

\subsection{Strict advantage of auxiliary memory}\label{subsec:memory_advantage}

We now address the question whether auxiliary memory provides an advantage even when the input probes are fully i.i.d.. Theorem~\ref{theo:unital_nogo} establishes that for every unital channel, the memoryless Stein exponent vanishes. Our next result shows that for every non-mixed-unitary channel, even if it is unital, the exponent becomes strictly positive with auxiliary memory. In fact, we get a straightforward lower bound in terms of the channel’s diamond-norm distance from the set of mixed-unitary channels.

We first recall the following channel-discrimination minimax identity \cite{yu2016bounds}.

\begin{lemma}[Theorem 1 in \cite{yu2016bounds}]\label{lemma:channel-minimax}
Let $\mathsf{S}_1$ and $\mathsf{S}_2$ be two nonempty convex compact sets of quantum channels on $\mathfrak{B}(\mathbb{C}^d)$. Then
\begin{align}
    \sup_{\rho_{RA}\in\mathfrak{D}(R\otimes\mathbb{C}^d)} \min_{\substack{\Phi_1\in\mathsf{S}_1\\\Phi_2\in\mathsf{S}_2}} \left\| \bigl(\mathrm{id}_R\otimes(\Phi_1-\Phi_2)\bigr)(\rho_{RA})\right\|_1 &= \min_{\substack{\Phi_1\in\mathsf{S}_1\\\Phi_2\in\mathsf{S}_2}} \max_{\rho_{RA}\in\mathfrak{D}(R\otimes\mathbb{C}^d)}\left\| \bigl(\mathrm{id}_R\otimes(\Phi_1-\Phi_2)\bigr)(\rho_{RA})\right\|_1\nonumber\\
    &=\min_{\substack{\Phi_1\in\mathsf{S}_1\\\Phi_2\in\mathsf{S}_2}} \|\Phi_1-\Phi_2\|_\diamond,
\end{align}
where $R\cong\mathbb{C}^d$.
\end{lemma}

\begin{theorem}\label{theo:memorypositive}
    For every non-mixed-unitary channel $\mathcal{N}$ on $\mathfrak{B}(\mathbb{C}^d)$
    \begin{align}
        \widetilde{E}^1\left(\mathsf{MU}^{\mathrm{iid}}\middle\|\mathcal{N}^{\mathrm{iid}}\right)\geq\frac{1}{2\ln2}\left[\min_{\Phi\in\mathsf{MU}}\left\|\Phi-\mathcal{N}\right\|_{\diamond}\right]^2>0.
    \end{align}
\end{theorem}

\begin{proof}
    Let $\rho_{RA}\in\mathfrak{D}(R\otimes \mathbb{C}^d)$ be an arbitrary bipartite probe state of which $A$ is used as the channel input and $R\cong\mathbb{C}^d$ is the auxiliary memory. Denoting the output as
    \begin{align*}
        \omega_{\mathcal{E},\rho}:=\mathrm{id}\otimes\mathcal{E}(\rho_{RA}),
    \end{align*}
    we have from the Pinsker's inequality 
    \begin{align}
        D\left(\omega_{\Phi,\rho}\middle\|\omega_{\mathcal{N},\rho}\right)\geq \frac{1}{2\ln2}\left\|\omega_{\Phi,\rho}-\omega_{\mathcal{N},\rho}\right\|^2_1\, ,
    \end{align}
    for every $\Phi\in\mathsf{MU}$ and $\mathcal{N}\notin\mathsf{MU}$.
    
    Using the characterization of the fully i.i.d., memory-assisted Stein exponent in Eq.~\eqref{eq:stein_mem_fully_iid}, we therefore get
    \begin{align}
        \widetilde{E}^1\left(\mathsf{MU}^{\mathrm{iid}}\middle\|\mathcal{N}^{\mathrm{iid}}\right)&\geq\frac{1}{2\ln2}\max_{\rho_{RA}\in\mathfrak{D}(R\otimes \mathbb{C}^d)}\min_{\Phi\in\mathsf{MU}}\left\|\omega_{\Phi,\rho}-\omega_{\mathcal{N},\rho}\right\|^2_1\nonumber\\
        &=\frac{1}{2\ln2}\left[\max_{\rho_{RA}\in\mathfrak{D}(R\otimes \mathbb{C}^d)}\min_{\Phi\in\mathsf{MU}}\left\|\omega_{\Phi,\rho}-\omega_{\mathcal{N},\rho}\right\|_1\right]^2.
    \end{align}
    The last line follows from the nonnegativity of the trace norm. Applying Lemma~\ref{lemma:channel-minimax} with $\mathsf{S}_1=\mathsf{MU}$ and $\mathsf{S}_2=\{\mathcal{N}\}$ (a singleton set), we get
    \begin{align}
        \widetilde{E}^1\left(\mathsf{MU}^{\mathrm{iid}}\middle\|\mathcal{N}^{\mathrm{iid}}\right)\geq\frac{1}{2\ln2}\left[\min_{\Phi\in\mathsf{MU}}\left\|\Phi-\mathcal{N}\right\|_\diamond\right]^2.
    \end{align}

    Finally, since the set $\mathsf{MU}$ is compact, the minimum diamond distance is achievable. If this minimum were zero, its minimizer would equal $\mathcal{N}$, contradicting $\mathcal{N}\notin\mathsf{MU}$. Therefore, the bound is strictly positive.
\end{proof}

Thus, every non-mixed-unitary channel is detectable with fully i.i.d. memory-assisted probes. Combined with Theorem \ref{theo:unital_nogo}, this establishes that auxiliary memory gives a strict advantage for every unital non-mixed-unitary channel, even with fully i.i.d. probes. Bounds on the diamond-norm distance between unital channels and $\mathsf{MU}$ have already been studied
independently in Ref.~\cite{yu2016bounds}.

\subsubsection{Entanglement between the channel-input and the auxiliary memory is not necessary}

It is interesting to note here that the origin of the advantage of the auxiliary memory cannot be solely attributed to the entanglement between the channel-input and the auxiliary memory. Our next result demonstrates that separable but correlated reference–input probes can also yield a strictly positive Stein exponent.

\begin{proposition}
    Consider the family of bipartite states in $\mathfrak{D}(\mathbb{C}^d\otimes\mathbb{C}^d)$
    \begin{align}
        \rho_{RA}(p)=p\,\ketbra{\Omega_d}{\Omega_d}+(1-p)\frac{\mathbb{I}_{RA}}{d^2} \in \mathfrak{D}(\mathbb{C}^d\otimes\mathbb{C}^d),\quad 0<p\leq1,
    \end{align}
    where $\ket{\Omega_d}=\frac{1}{\sqrt{d}}\sum_{i=0}^{d-1}\ket{ii}$. For every channel $\mathcal{N}\notin\mathsf{MU}$, the probe state $\rho_{RA}(p)$ yields a strictly positive fixed-probe Stein exponent:
    \begin{align}
        \min_{\Phi\in\mathsf{MU}} D\left(\mathrm{id}_R\otimes\Phi(\rho_{RA}(p))\middle\|\mathrm{id}_R\otimes\mathcal{N}(\rho_{RA}(p))\right)> 0,\qquad\forall\,p\in(0,1].
    \end{align}
    Moreover, $\rho_{RA}(p)$ is separable for $p\in(0,\frac{1}{d+1}]$.
\end{proposition}

\begin{proof}
    We prove this by contradiction. Suppose for some fixed value of $0<p\leq1$
    \begin{align}
        \min_{\Phi\in\mathsf{MU}} D\left(\mathrm{id}_R\otimes\Phi(\rho_{RA}(p))\middle\|\mathrm{id}_R\otimes\mathcal{N}(\rho_{RA}(p))\right)= 0.
    \end{align}
    That is, there exists a mixed-unitary channel $\Phi$ such that
    \begin{align}
        &\mathrm{id}_R\otimes\Phi(\rho_{RA}(p))=\mathrm{id}_R\otimes\mathcal{N}(\rho_{RA}(p)),\nonumber\\
        \text{or},\qquad &p\,J(\Phi)+(1-p)\frac{\mathbb{I}_R}{d}\otimes\Phi\left(\frac{\mathbb{I}_A}{d}\right)=p\,J(\mathcal{N})+(1-p)\frac{\mathbb{I}_R}{d}\otimes\mathcal{N}\left(\frac{\mathbb{I}_A}{d}\right),\label{eq:match}
    \end{align}
    where $J(\mathcal{X})$ denotes the normalised Choi-Jamiolkowski state of the channel $\mathcal{X}$. 
    
    Since $\Tr_R[\rho_{RA}(p)]=\frac{\mathbb{I}_A}{d}$, we must then have
    \begin{align}
        \Phi\left(\frac{\mathbb{I}_A}{d}\right)=\mathcal{N}\left(\frac{\mathbb{I}_A}{d}\right).
    \end{align}
    Then, Eq.~\eqref{eq:match} implies
    \begin{align}
        J(\Phi)=J(\mathcal{N}).
    \end{align}
    which is impossible as $\mathcal{N}\notin\mathsf{MU}$. Therefore, no such mixed-unitary channel exists, and our claim is established. The separability region of $\rho_{RA}(p)$ follows from \cite{PhysRevA.59.4206}.
\end{proof}

\subsubsection{Exact Stein exponent and optimal probes}

Thus far, we have focused on lower bounds for the Stein exponent. In general, determining its exact value and identifying optimal probes are difficult, even for a relatively simple discrimination strategy with fully i.i.d. probes with auxiliary memory. Nevertheless, for some specific channels, both can be determined. Our following result demonstrates this.

\begin{theorem}
    Let $d\geq3$ be odd and $\mathcal{N}$ be a quantum channel on $\mathfrak{B}(\mathbb{C}^d)$ that is unitarily equivalent to the $d$-dimensional Werner--Holevo channel $\mathcal{N}_{WH}$. Then
    \begin{align}
        \widetilde{E}^1\left(\mathsf{MU}^{\mathrm{iid}}\middle\|\mathcal{N}^{\mathrm{iid}}\right)=\infty,
    \end{align}
    and any full-Schmidt-rank pure state serves as an optimal probe.
\end{theorem}

\begin{proof}
    We prove the result for $\mathcal{N}_{WH}$ and the general statement then follows from the unitary equivalence. Consider an arbitrary full-Schmidt-rank pure probe $\ket{\psi}_{RA}\in\mathbb{C}^d\otimes\mathbb{C}^d$. Since it has full Schmidt rank, there exists an invertible operator $T$ \cite{Watrous2011} such that
    \begin{align}
        \ket{\psi}_{RA}=(T\otimes\mathbb{I}_A)\ket{\Omega_d},
    \end{align}
    where $\ket{\Omega_d}=\frac{1}{\sqrt{d}}\sum_{i=0}^{d-1}\ket{i}_R\ket{i}_A$. Consequently, for any channel $\mathcal{E}$, we get
    \begin{align}
        \mathrm{id}_R\otimes\mathcal{E}(\ketbra{\psi}{\psi})=(T\otimes\mathbb{I}_A)[J(\mathcal{E})](T^{\dagger}\otimes\mathbb{I}_A).\label{localinvert}
    \end{align}

    Now consider the normalised Choi-Jamiolkowski operator of the Werner--Holevo channel
    \begin{align}
        J(\mathcal{N}_{WH})&=\frac{1}{d-1}\left[\frac{\mathbb{I}_R\otimes\mathbb{I}_A}{d}-\Omega^{\intercal_{A}}_d\right]\nonumber\\
        &=\frac{2}{d(d-1)}P_{-},
    \end{align}
    where $P_{-}=(\mathbb I_{RA}-F)/{2}$ is the projector onto the antisymmetric subspace $\mathcal{A}_d$ of $\mathbb{C}^d\otimes \mathbb{C}^d$ with $F=\sum_{i,j=0}^{d-1}\ketbra{i}{j}_R\otimes \ketbra {j}{i}_A$ being the swap operator between the subsystems $R$ and $A$. Therefore,
    \begin{align}
        \operatorname{supp}\left(J(\mathcal{N}_{WH})\right)=\mathcal{A}_d.
    \end{align}

    On the other hand, the normalised Choi-Jamiolkowski operator of a mixed-unitary channel $\Phi$ is given by
    \begin{align}
        J(\Phi)=\sum_x p_x(\mathbb{I}\otimes U_x)\ketbra{\Omega_d}{\Omega_d}(\mathbb{I}\otimes U^{\dagger}_x).
    \end{align}
    Therefore,
    \begin{align}
        \operatorname{supp}\left(J(\Phi)\right)=\operatorname{span}\{(\mathbb{I}\otimes U_x)\ket{\Omega_d}:p_x>0\}.
    \end{align}
    For $(\mathbb{I}\otimes U_x)\ket{\Omega_d}$ to be an element of $\mathcal{A}_d$, we require
    \begin{align}
        U_x^{\intercal}=-U_x,\quad\forall x.
    \end{align}
    This implies
    \begin{align}
        \operatorname{det}(U_x)=\operatorname{det}(U^{\intercal}_x)=\operatorname{det}(-U_x)=(-1)^d\operatorname{det}(U_x).
    \end{align}
    Since $d$ is odd, we get
    \begin{align}
        \operatorname{det}(U_x)&=-\operatorname{det}(U_x),\nonumber\\
        \implies\operatorname{det}(U_x)&=0,
    \end{align}
    which contradicts the unitarity of $U_x$. Therefore,
    \begin{align}
        (\mathbb{I}\otimes U_x)\ket{\Omega_d}\notin\mathcal{A}_d,
    \end{align}
    and hence
    \begin{align}
        \operatorname{supp}\left(J(\Phi)\right)\not\subseteq \operatorname{supp}\left(J(\mathcal{N}_{WH})\right),\quad\forall\Phi\in\mathsf{MU}.
    \end{align}

    Let $K:=T\otimes\mathbb{I}_A$. For every positive semidefinite operator $X$,
    \begin{align}
        \operatorname{supp}(KXK^{\dagger})=K\operatorname{supp}(X).
    \end{align}
    Since $K$ is invertible, let two positive semidefinite operators be $X$ and $Y$ which satisfy
    \begin{align}
        \operatorname{supp}(KXK^\dagger)\subseteq \operatorname{supp}(KYK^\dagger), 
    \end{align}
    which means
    \begin{align}
        \nonumber&K\operatorname{supp}(X)\subseteq K \operatorname{supp}(Y)\\\implies&\operatorname{supp}(X)\subseteq  \operatorname{supp}(Y),
    \end{align}
    and the converse is also true. Hence,
    \begin{align}
        \operatorname{supp}(KXK^\dagger)\subseteq \operatorname{supp}(KYK^\dagger) \iff \operatorname{supp}(X)\subseteq \operatorname{supp}(Y).
    \end{align}
    Therefore, setting 
$X=J(\Phi),$ and $
Y=J(\mathcal N_{\rm WH})$, via Eq.~\eqref{localinvert} we have
    \begin{align}
        \operatorname{supp}\left[\mathrm{id}_R\otimes \Phi(\ketbra{\psi}{\psi})\right]\not\subseteq \operatorname{supp}\left[\mathrm{id}_R\otimes \mathcal{N}_{WH}(\ketbra{\psi}{\psi})\right],
    \end{align}    
    for every $\Phi\in\mathsf{MU}$. It then follows from the definition of quantum relative entropy that
    \begin{align}
        &D\left(\mathrm{id}_R\otimes \Phi(\ketbra{\psi}{\psi})\middle\|\mathrm{id}_R\otimes \mathcal{N}_{WH}(\ketbra{\psi}{\psi})\right)=\infty,\quad\forall\Phi\in\mathsf{MU},\nonumber\\
    \implies &\min_{\Phi\in\mathsf{MU}}D\left(\mathrm{id}_R\otimes \Phi(\ketbra{\psi}{\psi})\middle\|\mathrm{id}_R\otimes \mathcal{N}_{WH}(\ketbra{\psi}{\psi})\right)=\infty.
    \end{align}
    Therefore, the Stein exponent is given by
    \begin{align}
        \widetilde{E}^1\left(\mathsf{MU}^{\mathrm{iid}}\middle\|\mathcal{N}^{\mathrm{iid}}_{WH}\right)=\infty,
    \end{align}
    with every full-Schmidt-rank pure state $\ket{\psi}_{RA}$ being a maximizer.
\end{proof}

The infinite exponent operationally means that for every fixed full-Schmidt-rank pure probe, there exists a sequence of tests for which false-negative probability is exactly zero while the worst-case false-positive probability decreases exponentially. Thus, any prescribed nonzero false-positive error threshold $0<\epsilon<1$ can be met after sufficiently many channel uses while never mistaking the Werner-Holevo channel as a mixed-unitary channel. This, however, does not imply perfect discrimination with a single channel use.

\section{Conclusion}
\label{sec:conclu}

Quantum channel discrimination provides a fundamental framework for characterizing quantum devices, benchmarking their performance, and diagnosing noise that can limit quantum advantage. We studied the problem of determining whether an unknown quantum channel $\mathcal{E}$ is a specified non-mixed-unitary channel $\mathcal{N}$ or belongs to the set of mixed-unitary channels, given multiple uses of the channel. We formulated this as a composite quantum channel hypothesis testing task, with the set of mixed-unitary channels $\mathsf{MU}$ as the composite i.i.d. null hypothesis and the specified channel $\mathcal{N}\notin\mathsf{MU}$ as the simple i.i.d. alternative hypothesis. We focused on parallel channel discrimination strategies and investigated how various operational restrictions on the input probes affect the achievable Stein exponents.

In general composite quantum channel hypothesis testing, arbitrary multipartite probes are  typically used as channel inputs, leading to regularized Stein exponents, while fully i.i.d. probes yield finite-letter expressions by restricting the admissible inputs. As an intermediate between these two regimes, we therefore introduced
the block-i.i.d. probe setting  in which block size one recovers fully i.i.d. probes and in the limit of infinite block size approaches arbitrary probes; for every fixed block size, we derived finite-letter characterizations of both memoryless and memory-assisted Stein exponents.  We further classified strategies according to the allowed intra-block correlations, including product, fully separable, and entangled probes, yielding a multifaceted hierarchy. 
Applying this framework to non-mixed-unitarity testing, we established that the resulting exponents are monotonic under mixed-unitary pre- and post-processing and invariant under unitary equivalence. Moreover, we proved several strict separations within the hierarchy for unital non-mixed-unitary channels: (1) without auxiliary memory, fully i.i.d. probes yield a vanishing Stein exponent for every unital channel, even though the set of unital channels is strictly larger than the set of mixed unitary channels for dimensions \(d\geq 3\); (2) for the qutrit Werner–Holevo channel, this limitation is overcome with three appropriately chosen product probes, yielding a strictly positive exponent; and (3) for block size two, we proved that pure product qutrit probes and several broad classes of mixed-state pairs remain insufficient, with numerical and analytical evidence suggesting the same for arbitrary product probe pairs. We formulated this last observation as a conjecture. Proving this conjecture would imply that three is the minimum number of uncorrelated inputs required to detect the non-mixed-unitarity of the qutrit Werner-Holevo channel, and, therefore, is an interesting avenue for future research. Another open question that has not been answered in this work is whether classical correlation between two input probes can yield nonzero Stein exponent for the qutrit Werner-Holevo channel.

An entangled two-qutrit probe yields a strictly positive exponent for the Werner–Holevo channel, even though two uncorrelated probes appear insufficient. In contrast, every maximally entangled two-qutrit probe gives a vanishing exponent.  These findings highlight that while entanglement between channel uses can enhance discrimination, the advantage is sensitive to the structure of the probe and does not increase monotonically with its entanglement. Access to auxiliary memory changes the situation more fundamentally. We proved that fully i.i.d. memory-assisted probes yield a strictly positive Stein exponent for every non-mixed-unitary channel and  obtained a universal lower bound in terms of the channel's diamond-norm distance from the mixed-unitary set. Moreover, this benefit does not require entanglement between the channel input and the reference: a universal family of separable but correlated probes already suffices. For Werner–Holevo channels in odd dimensions, the separation becomes maximal in the asymptotic sense — the memory-assisted Stein exponent is infinite, and every full-Schmidt-rank pure state is an optimal probe. Our results reveal that even the parallel-strategy framework, presumably simplest of all existing strategies, already exhibits highly nontrivial structure in detecting non-mixed-unitarity, thereby motivating a systematic investigation of this problem within more general adaptive strategies. 

 \section*{acknowledgements}
The authors used OpenAI’s ChatGPT (GPT-5.6 Sol) for verifying few proofs of the work. The authors take full responsibility for every claim, calculation, and conclusion.

We acknowledge support from the project entitled ``Technology Vertical - Quantum Communication'' under the National Quantum Mission of the Department of Science and Technology (DST)  (Sanction Order No. DST/QTC/NQM/QComm/2024/2 (G)).

\appendix

\section{Alternate proof of Proposition~\ref{prop:useless_probes} for pure states}\label{appa}
Before presenting the main result let us recall a useful lemma.

\begin{lemma}[\cite{CHEFLES_2004}]\label{Gram}
Let$\{\ket{\psi_1},\cdots,\ket{\psi_n}\}$ and $\{\ket{\phi_1},\cdots,\ket{\phi_n}\}$ be two sets of pure states. Then there exists a unitary operator \(U\) such that
\[
U\ket{\psi_i}=\ket{\phi_i},
\qquad i=1,\cdots,n,
\]
\emph{if and only if} the two sets have identical Gram matrices, i.e.,
\begin{align}
    G^{\psi}=G^{\phi},
\end{align}
where $G^{\psi}_{ij}:=\braket{\psi_i}{\psi_j}$ and $G^{\phi}_{ij}:=\braket{\phi_i}{\phi_j}$.
\end{lemma}

\begin{proposition}
Let $\mathcal{N}$ be an $O(3)$-covariant channel. Then for every pair of pure states $\{\ket{\psi_1},\ket{\psi_2}\}\subset \mathbb{C}^3$, there exists a mixed-unitary channel $\Phi$ such that 
\[
\Phi(\ketbra{\psi_i}{\psi_i}) = \mathcal{N}(\ketbra{\psi_i}{\psi_i}),
\quad \mbox{for}~i=1,2.
\]
\end{proposition}

\begin{proof}
We first establish this claim for the Werner-Holevo channel $\mathcal{N}_{WH}$. We then extend the result to arbitrary $O(3)$-covariant non-mixed-unitary unital channels using the convex structure of the set of $O(3)$-covariant channels.

\noindent \textbf{For Werner-Holevo channel:} Let $\ket{\psi_1}$ and $\ket{\psi_2}$ be any two distinct qutrit pure states, and let $\alpha:=\braket{\psi_1}{\psi_2}$ denote their inner product. Since $\ket{\psi_1}$ and $\ket{\psi_2}$ are distinct, they are linearly independent, and hence $0\le \abs{\alpha}<1$.

Define the two-dimensional subspaces
\begin{align}
S_i:=\operatorname{supp}\!\left(\frac{1}{2}(\mathbb{I}_3-\ketbra{\psi_i^*}{\psi_i^*})\right)=(\operatorname{span}\{\ket{\psi_i^*}\})^\perp,\quad i=1,2.    
\end{align}
Since \(S_1\) and \(S_2\) are distinct two-dimensional subspaces of \(\mathbb{C}^3\), their intersection is one-dimensional:
\[
\dim(S_1\cap S_2)=1.
\]
Let \(\ket{\eta}\) be a unit vector spanning this intersection, \emph{i.e.},
\begin{align}
S_1\cap S_2=\operatorname{span}\{\ket{\eta}\}.
\end{align}
For a fixed pair \(\{\ket{\psi_1},\ket{\psi_2}\}\), the vector \(\ket{\eta}\) is uniquely determined up to a global phase. Next, we choose unit vectors \(\ket{\xi_i}\in S_i\) such that
\begin{align}
\braket{\xi_i}{\eta}=0,\quad i=1,2.
\end{align}
Then, for each \(i=1,2\), the set $\{\ket{\xi_i},\ket{\eta}\}$ forms an orthonormal basis of \(S_i\). Here also, for a fixed pair \(\{\ket{\psi_1},\ket{\psi_2}\}\), the vectors \(\ket{\xi_1}\) and \(\ket{\xi_2}\) are uniquely determined up to independent global phases.

Consequently,
\begin{align}\label{decom}
\frac{1}{2}\left(\mathbb{I}_3-\ketbra{\psi_i^*}{\psi_i^*}\right)=\frac{1}{2}\Bigl(\ketbra{\xi_i}{\xi_i}+\ketbra{\eta}{\eta}\Bigr),\qquad i=1,2.
\end{align}
Furthermore, Eq.~\eqref{decom} implies that
\begin{align*}
\abs{\braket{\xi_1}{\xi_2}}=\abs{\braket{\psi_1^*}{\psi_2^*}}=\abs{\braket{\psi_1}{\psi_2}}=\abs{\alpha}.
\end{align*}
Since \(\ket{\xi_1}\) and \(\ket{\xi_2}\) are defined only up to independent global phases, we may choose those phases such that
\begin{align}\label{inner}
\braket{\xi_1}{\xi_2}=\braket{\psi_1}{\psi_2}=\alpha.
\end{align}

To prove the existence of a mixed-unitary channel \(\Phi\) satisfying
\begin{align}
\Phi(\ketbra{\psi_i}{\psi_i}) &= \mathcal{N}_{WH}(\ketbra{\psi_i}{\psi_i})=\frac{1}{2}\Bigl(\mathbb{I}_3-\ketbra{\psi_i^*}{\psi_i^*}\Bigr),\qquad i=1,2,
\end{align}
we need to show that there exist unitary operators \(\{U_k\}\) and a probability distribution \(\{p_k\}\) such that
\begin{align}
U_k\ket{\psi_i} &= \ket{u_{i,k}}\in S_i, \qquad \forall\,k,\ i=1,2,\label{cond1}\\
\sum_k p_k\ketbra{u_{i,k}}{u_{i,k}} &\overset{\eqref{decom}}{=} \frac{1}{2}\Bigl(\ketbra{\xi_i}{\xi_i}+\ketbra{\eta}{\eta}\Bigr),\qquad i=1,2.\label{cond2}
\end{align}

By Lemma~\ref{Gram}, Eq.~\eqref{cond1} holds if and only if
\begin{align}
\braket{u_{1,k}}{u_{2,k}}=\braket{\psi_1}{\psi_2}=\alpha,\qquad \forall\,k.
\end{align}
Furthermore, any normalized vector \(\ket{u_{i,k}}\in S_i\) can be expressed in the basis $\{\ket{\xi_i},\ket{\eta}\}$ as
\begin{align}
\ket{u_{i,k}}=\sqrt{r_{i,k}}\ket{\xi_i}+\sqrt{1-r_{i,k}}\,e^{\iota \phi_{i,k}}\ket{\eta},\label{rep}
\end{align}
where
$r_{i,k}\in[0,1]$ and $\phi_{i,k}\in[0,2\pi)$, for all \(k\) and \(i=1,2\). Therefore, Eq.~\eqref{cond1} can equivalently be written as
\begin{align}
&\alpha = \braket{u_{1,k}}{u_{2,k}} \overset{\eqref{rep}}{=} \sqrt{r_{1,k}r_{2,k}}\braket{\xi_1}{\xi_2}+\sqrt{(1-r_{1,k})(1-r_{2,k})}\,e^{\iota(\phi_{2,k}-\phi_{1,k})},\qquad\forall k,\nonumber\\
\text{or,}~~&\alpha(1-\sqrt{r_{1,k}r_{2,k}})\overset{\eqref{inner}}{=}\sqrt{(1-r_{1,k})(1-r_{2,k})}\,e^{\iota(\phi_{2,k}-\phi_{1,k})},\qquad\forall k.\label{cond1'}
\end{align}
Eq.~\eqref{cond2} can equivalently be expressed in terms of
\(r_{i,k}\) and \(\phi_{i,k}\) as
\begin{align}
\left.
\begin{aligned}
&\sum_k p_k r_{i,k}=\frac{1}{2},\\
&\sum_k p_k \sqrt{r_{i,k}(1-r_{i,k})}\,e^{\iota \phi_{i,k}}=0
\end{aligned}
\right\}
\qquad \text{for } i=1,2.\label{cond2'}
\end{align}

We now show that for every $0\leq\abs{\alpha}<1$, there exist $\{r_{i,k},\phi_{i,k}\}$ and $\{p_k\}$ satisfying Eqs.~\eqref{cond1'} and \eqref{cond2'}. 

Let $k\in\{1,2\}$ and choose
\begin{subequations}
\begin{align}
&r_{1,1}=1-r_{1,2}=x,\\
&r_{2,1}=1-r_{2,2}=y,\\
&\phi_{i,2}=\phi_{i,1}+\pi,\quad i=1,2,\\
&\phi_{2,k}-\phi_{1,k}=\arg{(\alpha)},\quad k=1,2,\\
&\phi_{1,1}=\gamma,\\
&p_1=p_2=\frac{1}{2},
\end{align}    
\end{subequations}
where \(x,y\in[0,1]\) and \(\gamma\in[0,2\pi)\).

With these choices, Eq.~\eqref{cond2'} is automatically satisfied for all admissible values of $x$, $y$, and $\gamma$. Meanwhile, Eq.~\eqref{cond1'} reduces to
\begin{align}
    &\abs{\alpha}(1-\sqrt{xy})=\sqrt{(1-x)(1-y)},\quad(\text{for}~k=1),\nonumber\\
    \text{and}~~&\abs{\alpha}\left(1-\sqrt{(1-x)(1-y)}\right)=\sqrt{xy},\quad(\text{for}~k=2).\nonumber
\end{align}
Since $\abs{\alpha}<1$, these equations imply
\begin{align}
&\sqrt{xy}=\frac{|\alpha|}{1+|\alpha|},\label{hyp}\\
\text{and}~~&x+y=1.\label{lin}
\end{align}
For every $0\le \abs{\alpha}<1$, Eqs.~\eqref{hyp} and \eqref{lin} admit the solutions 
\begin{align}
    &x=\frac{1}{2}\left(1\pm\sqrt{1-\frac{4|\alpha|^2}{(1+|\alpha|)^2}}\right),\label{sol}\\
    \text{and}~~&y=1-x,
\end{align}
which lie in the interval $[0,1]$.

Therefore, for every pair of pure states $\{\ket{\psi_1},\ket{\psi_2}\}$, we have a mixed-unitary channel
\begin{align*}
    \Phi(X)=\frac{1}{2}\left(U_1XU_1^{\dagger}+U_2XU_2^{\dagger}\right),
\end{align*}
where the unitary operators \(U_1,U_2\) satisfy
\begin{align*}
    &U_1\ket{\psi_1}=\sqrt{x}\ket{\xi_1}+\sqrt{1-x}\,e^{\iota \gamma}\ket{\eta},\\
    &U_2\ket{\psi_1}=\sqrt{1-x}\ket{\xi_1}-\sqrt{x}\,e^{\iota \gamma}\ket{\eta},\\
    &U_1\ket{\psi_2}=\sqrt{1-x}\ket{\xi_2}+\sqrt{x}\,e^{\iota (\gamma+\arg{(\alpha)})}\ket{\eta},\\
    &U_2\ket{\psi_2}=\sqrt{x}\ket{\xi_2}-\sqrt{1-x}\,e^{\iota (\gamma+\arg{(\alpha)})}\ket{\eta},
\end{align*}
with $x$ given by Eq.~\eqref{sol} and $\gamma\in[0,2\pi)$, such that
\begin{align*}
\Phi(\ketbra{\psi_i}{\psi_i})=\frac{1}{2}\Bigl(\ketbra{\xi_i}{\xi_i}+\ketbra{\eta}{\eta}\Bigr)=\frac{1}{2}\left(\mathbb{I}_3-\ketbra{\psi_i^*}{\psi_i^*}\right),\qquad i=1,2.
\end{align*}

\noindent \textbf{For arbitrary $O(3)$-covariant channels:} Let $\mathcal{N}$ be an arbitrary $O(3)$-covariant non-mixed-unitary unital channel. Then $\mathcal{N}$ can be expressed as
\begin{align*}
\mathcal{N}=q_1 \mathcal{N}_{WH}+q_2 \Psi+q_3 \mathrm{id},
\end{align*}
where $\mathrm{id},\Psi\in\mathsf{MU}$, and
$\{q_i\}$ is a probability vector.

For a given pair of pure states $\{\ket{\psi_1},\ket{\psi_2}\}$, let $\Phi$ be a mixed-unitary channel satisfying $$\Phi(\ketbra{\psi_i}{\psi_i})=\mathcal{N}_{WH}(\ketbra{\psi_i}{\psi_i}),\qquad i=1,2.$$
Define the channel
\begin{align*}
\widetilde{\Phi}=q_1 \Phi+q_2 \Psi+q_3 \mathrm{id}.
\end{align*}
Since the set of mixed-unitary channels is convex, \(\widetilde{\Phi}\) is also mixed-unitary. Moreover,
\begin{align}
\widetilde{\Phi}(\ketbra{\psi_i}{\psi_i})=\mathcal{N}(\ketbra{\psi_i}{\psi_i}),\qquad i=1,2.
\end{align}

Therefore, for every pair of pure states $\{\ket{\psi_1},\ket{\psi_2}\}$, there exists a mixed-unitary channel $\widetilde{\Phi}$ that reproduces the action of $\mathcal{N}$ on both states. This completes the proof.
\end{proof}

\bibliography{bib}
\end{document}